\documentclass[11pt,onecolumn]{IEEEtran}

\usepackage[mathscr]{eucal}
\usepackage[cmex10]{amsmath}
\usepackage{epsfig,epsf,psfrag}
\usepackage{amssymb,amsmath,amsthm,amsfonts,latexsym}
\usepackage{amsmath,graphicx,bm,xcolor,url}
\usepackage[caption=false]{subfig} 
\usepackage{fixltx2e}%ordering of single and double column floats
\usepackage{array}%array and tabular environments
\usepackage{verbatim}
\usepackage{bm}
\usepackage{algorithmic, cite}
\usepackage{algorithm}
\usepackage{verbatim}
\usepackage{textcomp}
\usepackage{mathrsfs}
\usepackage{epstopdf}
\catcode`~=11 \def\UrlSpecials{\do\~{\kern -.15em\lower .7ex\hbox{~}\kern .04em}} \catcode`~=13 

\allowdisplaybreaks[1]
 
\newcommand{\nn}{\nonumber}

\newcommand{\calA}{\mathcal{A}}

\newcommand{\calC}{\mathcal{C}}
\newcommand{\calD}{\mathcal{D}}
\newcommand{\calE}{\mathcal{E}}
\newcommand{\calF}{\mathcal{F}}

\newcommand{\calL}{\mathcal{L}}
\newcommand{\calM}{\mathcal{M}}
\newcommand{\calN}{\mathcal{N}}
\newcommand{\calO}{\mathcal{O}}

\newcommand{\calS}{\mathcal{S}}

\newcommand{\calV}{\mathcal{V}}

\newcommand{\bV}{\mathbf{V}}

\newcommand{\rmb}{\mathrm{b}}

\newcommand{\rmf}{\mathrm{f}}

\newcommand{\rmi}{\mathrm{i}}

\newcommand{\rmP}{\mathrm{P}}

\newcommand{\bbE}{\mathbb{E}}

\newcommand{\bbN}{\mathbb{N}}

\newcommand{\bbR}{\mathbb{R}}

\DeclareMathAlphabet{\mathbsf}{OT1}{cmss}{bx}{n}
\DeclareMathAlphabet{\mathssf}{OT1}{cmss}{m}{sl}% slanted sans serif

\newcommand{\rvC}{\mathsf{C}}

\newcommand{\rvV}{\mathsf{V}}

\newcommand{\rvW}{\mathsf{W}}

\DeclareSymbolFont{bsfletters}{OT1}{cmss}{bx}{n}  
\DeclareSymbolFont{ssfletters}{OT1}{cmss}{m}{n}
\DeclareMathSymbol{\bsfGamma}{0}{bsfletters}{'000}
\DeclareMathSymbol{\ssfGamma}{0}{ssfletters}{'000}
\DeclareMathSymbol{\bsfDelta}{0}{bsfletters}{'001}
\DeclareMathSymbol{\ssfDelta}{0}{ssfletters}{'001}
\DeclareMathSymbol{\bsfTheta}{0}{bsfletters}{'002}
\DeclareMathSymbol{\ssfTheta}{0}{ssfletters}{'002}
\DeclareMathSymbol{\bsfLambda}{0}{bsfletters}{'003}
\DeclareMathSymbol{\ssfLambda}{0}{ssfletters}{'003}
\DeclareMathSymbol{\bsfXi}{0}{bsfletters}{'004}
\DeclareMathSymbol{\ssfXi}{0}{ssfletters}{'004}
\DeclareMathSymbol{\bsfPi}{0}{bsfletters}{'005}
\DeclareMathSymbol{\ssfPi}{0}{ssfletters}{'005}
\DeclareMathSymbol{\bsfSigma}{0}{bsfletters}{'006}
\DeclareMathSymbol{\ssfSigma}{0}{ssfletters}{'006}
\DeclareMathSymbol{\bsfUpsilon}{0}{bsfletters}{'007}
\DeclareMathSymbol{\ssfUpsilon}{0}{ssfletters}{'007}
\DeclareMathSymbol{\bsfPhi}{0}{bsfletters}{'010}
\DeclareMathSymbol{\ssfPhi}{0}{ssfletters}{'010}
\DeclareMathSymbol{\bsfPsi}{0}{bsfletters}{'011}
\DeclareMathSymbol{\ssfPsi}{0}{ssfletters}{'011}
\DeclareMathSymbol{\bsfOmega}{0}{bsfletters}{'012}
\DeclareMathSymbol{\ssfOmega}{0}{ssfletters}{'012}

\newcommand{\hatb}{\hat{b}}

\newcommand{\hatm}{\hat{m}}

\newcommand{\tiln}{\tilde{n}}

\newcommand{\hatP}{\hat{P}}

\newcommand{\hatW}{\hat{W}}

\newcommand{\tilx}{\tilde{x}}

\newcommand{\tilY}{\tilde{Y}}

\newcommand{\bara}{\bar{a}}

\newcommand{\barx}{\bar{x}}

\newcommand{\eps}{\varepsilon}
\DeclareMathOperator{\tr}{tr}

\newcommand{\bzero}{\mathbf{0}}

\newtheorem{theorem}{Theorem} 
\newtheorem{lemma}[theorem]{Lemma}

\newtheorem{proposition}[theorem]{Proposition}
\newtheorem{corollary}[theorem]{Corollary}
\newtheorem{definition}{Definition}

\newcommand{\qednew}{\nobreak \ifvmode \relax \else
      \ifdim\lastskip<1.5em \hskip-\lastskip
      \hskip1.5em plus0em minus0.5em \fi \nobreak
      \vrule height0.75em width0.5em depth0.25em\fi}

\newcommand{\rc}{\mathrm{rc}}
\usepackage{epsfig} 
\usepackage{epstopdf,bbm} 
\usepackage{hyperref}
\usepackage{overpic}

\usepackage{graphicx}
\usepackage{tikz}
\usetikzlibrary{arrows.meta,patterns}

\usepackage{amsmath,amssymb}
\usepackage{xcolor}
\usetikzlibrary{arrows.meta,calc,positioning,decorations.pathreplacing}

\definecolor{lowerr}{HTML}{D9EEF5}
\definecolor{miderr}{HTML}{F8D49A}
\definecolor{higherr}{HTML}{D96B5F}
\definecolor{keepgreen}{HTML}{2F7F62}
\definecolor{trimgray}{HTML}{D8D8D8}
\definecolor{textgray}{HTML}{555555}

\tikzset{
  stagearrow/.style={-{Latex[length=3.2mm,width=2.2mm]}, very thick, draw=textgray},
  paneltitle/.style={font=\bfseries\small, align=center},
  annotation/.style={font=\scriptsize, align=center, text width=4.35cm},
  axislabel/.style={font=\scriptsize, text=textgray},
  matrixoutline/.style={draw=black!75, line width=0.55pt},
}

\begin{document}
\flushbottom
\title{Second-Order Asymptotics for the  Gaussian Multiple-Access Channel at  Corner Points} 

\author{Vincent Y.~F.~Tan  %
\thanks{V.~Y.~F.~Tan is with the Department of Mathematics and 
the Department of Electrical and Computer Engineering, National University of Singapore
(email:\,vtan@nus.edu.sg).
}}
 
\IEEEpeerreviewmaketitle

\maketitle
%\vspace{-.4in}
\begin{abstract}
We establish  exact second-order coding rate regions at the two corner points of the capacity region of the two-user Gaussian multiple-access channel. For any  average error probability $\eps\in(0,1)$, we characterize the $n^{-1/2}$-scale fluctuations of achievable rates around each corner point, proving a converse that matches the previously known achievability bound. The proof first extracts a rectangular subcode that preserves the independence of the two transmitted codewords. The Polyanskiy--Verd\'u good-code output distribution theorem then yields log-determinant constraints that induce a spectral decomposition of a trimmed codebook into diffuse and exceptional subspaces. An entropic Brascamp--Lieb projection inequality controls the codeword inner product on the diffuse subspace, while variance-inflated Gaussian output distributions handle the low-dimensional exceptional subspace. Combining these two treatments yields the joint Gaussian limit required for the matching second-order converse.
\end{abstract}

 
\section{Introduction}
The multiple-access channel (MAC) is one of the most fundamental communication  models in multi-user information theory. It describes communication from multiple transmitters to a common receiver, as in the uplink of a wireless communication system. Shannon~\cite{Shannon1961TwoWay} initiated the systematic study of multiway communication channels, and Ahlswede~\cite{Ahlswede1971Multiway,Ahlswede1974TwoSenders} and Liao~\cite{Liao1972MultipleAccess} subsequently established the single-letter capacity region of the discrete memoryless MAC.

In this paper, we revisit the two-user Gaussian MAC, in which the sum of the two channel inputs is corrupted by additive white Gaussian noise (AWGN). Cover~\cite{Cover1975SomeAdvances} and Wyner~\cite{wyner74} showed that its capacity region, illustrated in Fig.~\ref{fig:capacity-scope}, is the pentagon that consists of all rate pairs $(R_1,R_2)$ satisfying
\begin{align}
R_1&\le \rvC(P_1),\qquad
R_2\le \rvC(P_2),\qquad
R_1+R_2\le \rvC(P_1+P_2),
\label{eqn:cap_region_intro}
\end{align}
where $P_j$ is the maximal per-codeword power constraint for user $j \in \{1,2\}$ and $\rvC(x):=\frac12\log(1+x)$  is the Gaussian capacity function. Achievability is established using
independent Gaussian random codebooks, while the weak converse follows from
Fano's inequality and standard single-letterization arguments~\cite{elgamal}.

Although this first-order capacity region has been known for more than five decades, it does not describe how rapidly the rates attainable at finite blocklength approach its boundary. For a code of blocklength $n$, the relevant rate deviations typically occur on the scale $n^{-1/2}$. The study of these deviations is known as \emph{second-order asymptotics}~\cite{Tan14}. Second-order inner bounds for the Gaussian MAC were obtained by MolavianJazi and Laneman~\cite{MolavianJaziLaneman2015}, Scarlett, Martinez,
and Guill\'en i F\`abregas~\cite{ScarlettMartinezGuillen2015}, and Yavas, Kostina, and Effros~\cite{yavas21}. On the converse side, Fong and Tan~\cite{fongtan16} confined finite-blocklength rate pairs to an \(O(\sqrt{\log n/n})\) enlargement of the capacity region. Kosut~\cite{Kosut2022Wringing} subsequently sharpened this to \(O(n^{-1/2})\), providing the best general outer bound and establishing the correct second-order scale. Kosut's bound, however, does not determine the exact second-order coefficients or the corresponding second-order coding rate region. The purpose of this paper is to characterize the second-order coding rate
region exactly at the two corner points of the  capacity
region of the two-user Gaussian MAC. We do not treat the relative interior of the sum-rate face: there the
sum-rate constraint is the only active first-order constraint, whereas our
converse relies essentially on an active individual-rate constraint to obtain
the codebook regularity underlying the spectral decomposition. This obstruction
is discussed in detail in Section~\ref{sec:concl}.

%The purpose of this paper is to identify this region exactly at the two corner points of the two-user Gaussian MAC capacity region; the treatment of the relative interior of the sum-rate face requires rather different techniques (see Section~\ref{sec:concl} for a detailed discussion).

% Manuscript-ready TikZ figure.
%
% Required in the main manuscript preamble:
%   \usepackage{tikz}
%
% No additional TikZ libraries are required.  The figure uses the manuscript
% macros \rvC and P_\Sigma.

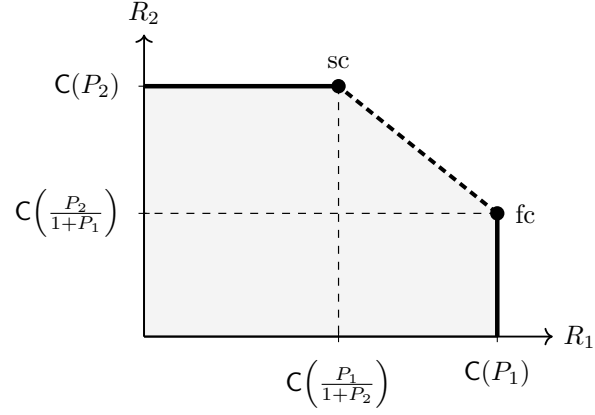
\begin{figure}[t]
    \centering
    \begin{tikzpicture}[
        x=1.05cm,
        y=1.05cm,
        every node/.style={font=\small},
        axis/.style={->,line width=0.7pt},
        boundary/.style={line width=0.7pt},
        knownface/.style={line width=1.7pt},
        openface/.style={line width=1.5pt, densely dashed},
        guide/.style={line width=0.45pt,dashed}
    ]

        % Schematic coordinates; only their ordering is significant.
        \coordinate (O) at (0,0);
        \coordinate (R) at (4.45,0);
        \coordinate (A) at (4.45,1.55);
        \coordinate (B) at (2.45,3.15);
        \coordinate (T) at (0,3.15);

        % Capacity region and its boundary.
        \fill[gray!9] (O)--(R)--(A)--(B)--(T)--cycle;
        \draw[boundary] (O)--(R)--(A);
        \draw[boundary] (O)--(T)--(B);

        % Boundary portions with a single active individual-rate constraint.
        \draw[knownface] (R)--(A);
        \draw[knownface] (B)--(T);

        % Relative interior of the sum-rate face.
        \draw[openface] (A)--(B);

        % Coordinate guides for the two corner points.
        \draw[guide] (B)--(2.45,0);
        \draw[guide] (A)--(0,1.55);

        % The two corner points treated in the paper.
        \fill (A) circle (2.7pt) node[right=3pt] {$\mathrm{fc}$};
        \fill (B) circle (2.7pt) node[above=3pt] {$\mathrm{sc}$};

        % Axes and the two individual-rate limits.
        \draw[axis] (0,0)--(5.15,0) node[right] {$R_1$};
        \draw[axis] (0,0)--(0,3.80) node[above] {$R_2$};
        \draw (4.45,0.07)--(4.45,-0.07)
            node[below=3pt] {$\rvC(P_1)$};
        \draw (2.45,0.07)--(2.45,-0.07)
            node[below=3pt]
            {$\rvC\!\left(\frac{P_1}{1+P_2}\right)$};
        \draw (0.07,3.15)--(-0.07,3.15)
            node[left=3pt] {$\rvC(P_2)$};
        \draw (0.07,1.55)--(-0.07,1.55)
            node[left=3pt]
            {$\rvC\!\left(\frac{P_2}{1+P_1}\right)$};

    \end{tikzpicture}
    \caption{The Gaussian MAC capacity region.  The filled circles are the
    first-corner point
    $\mathrm{fc}=(\rvC(P_1),\rvC(P_\Sigma)-\rvC(P_1))$ and second-corner point
    $\mathrm{sc}=(\rvC(P_\Sigma)-\rvC(P_2),\rvC(P_2))$, whose exact second-order
    coding rate regions are characterized in this paper.  The bold vertical
    and horizontal segments have a single active individual-rate constraint
    and hence scalar second-order behavior~\cite{Hayashi09,PPV10}. The heavy dashed segment
    joining $\mathrm{fc}$ and $\mathrm{sc}$ is the relative interior of the sum-rate face, for
    which the exact second-order region remains open.  }
    \label{fig:capacity-scope}
\end{figure}
%This
%does not resolve the $O(n^{-1/2})$ scale achieved by the existing inner bounds, leaving a factor of $\sqrt{\log n}$ between the achievability and converse estimates.
\subsection{Related Work}
For the point-to-point AWGN channel under a maximal per-codeword power constraint
$P$, Hayashi~\cite{Hayashi09} and Polyanskiy, Poor, and
Verd\'u~\cite{PPV10} established the normal approximation
\begin{align}
\log M^*(n,\eps,P)
=
n\rvC(P)
+\sqrt{n\rvV(P)}\,\Phi^{-1}(\eps)
+O(\log n),
\end{align}
where $M^*(n,\eps,P)$ is the largest code size achievable with average error
probability at most $\eps$,
$ \rvV(x):=\frac{x(x+2)}{2(x+1)^2}$ 
is the AWGN dispersion, and $\Phi^{-1}$ is the inverse CDF of a standard
Gaussian random variable. Polyanskiy, Poor, Verd\'u~\cite[Thm.~54]{PPV10} and Tan and Tomamichel \cite{TanTom15} identified the third-order term and showed the refined asymptotic expansion
\begin{align}
\log M^*(n,\eps,P)
=
n\rvC(P)
+\sqrt{n\rvV(P)}\,\Phi^{-1}(\eps)
+\frac12\log n+O(1).
\end{align}

An exact second-order characterization was previously obtained by Scarlett
and Tan~\cite{ST2015} for the Gaussian MAC with {\em  degraded
message sets}, in which one encoder knows both messages. Although the messages
are independent, the informed encoder transmits
$x_1^n(m_1,m_2)$, so its codeword may be correlated with
$x_2^n(m_2)$. This additional flexibility is crucial in their converse: for
each $m_2$, the user-1 codewords can be relabelled and expurgated separately,
allowing the codeword pairs to be reduced to a common empirical-correlation
class. Such a reduction is not available for the ordinary MAC, where the
codebooks have the Cartesian product form
$\{x_1^n(m_1)\}\times\{x_2^n(m_2)\}$ and the two message-uniform codewords
are independent. Pair-dependent expurgation would generally destroy this
product structure. The difficulty in the ordinary Gaussian MAC is therefore
to preserve codeword independence while controlling their (random) inner
product.

Since our objective is a second-order converse for the Gaussian MAC, we  
review the wringing-based line of MAC  converse arguments most closely related
to the present work.  Dueck~\cite{Dueck1981StrongConverse} proved the strong
converse for the two-user discrete memoryless MAC by combining expurgation
and wringing with the blowing-up method of Ahlswede, G\'acs, and
K\"orner~\cite{AhlswedeGacsKorner1976}, whose concentration-of-measure
origins can be traced to Margulis~\cite{Margulis1974Probabilistic}.
Retaining only message pairs with small conditional error generally destroys
the Cartesian product structure of the message set and thereby introduces (unwanted)
dependence between the transmitted codewords. Wringing extracts a
sufficiently large subcode on which this dependence can be controlled.
Ahlswede~\cite{Ahlswede1982Elementary} later gave an alternative proof by
refining Dueck's wringing argument and replacing the blowing-up step with
Augustin's nonasymptotic converse bound~\cite{Augustin1966}. For  finite
alphabets, tracking the finite blocklength terms in Ahlswede's proof gives an
$O((\log n)/\sqrt n)$ enlargement of the capacity region in normalized
rates, equivalently an $O(\sqrt n\log n)$ remainder in log code size. These
classical arguments do not extend directly to the Gaussian MAC: Dueck's
combinatorial list bound uses the finiteness of the output alphabet, whereas
the constants in Ahlswede's argument depend explicitly on the input and
output alphabet cardinalities and do not readily accommodate power
constraints. Han's information spectrum framework provides more general
sufficient conditions for a strong converse~\cite{Han98}, but these
conditions can be difficult to verify for concrete classes of memoryless
MACs.

Fong and Tan~\cite{fongtan16} addressed the continuous-alphabet difficulty by
quantizing the Gaussian inputs, applying Ahlswede-style wringing to the
resulting discrete variables, and invoking the hypothesis testing
converse of Wang, Colbeck, and Renner~\cite{WangColbeckRenner2009}. Their proof establishes the strong converse for the
Gaussian MAC and gives an $O(\sqrt{\log n/n})$ enlargement of the capacity
region in normalized rates, equivalently an $O(\sqrt{n\log n})$ remainder in
log code size. In a breakthrough work, Kosut~\cite{Kosut2022Wringing} subsequently eliminated the
extra $\sqrt{\log n}$ factor by introducing \emph{wringing dependence}, a
measure designed to quantify the residual dependence created by expurgation.
Under regularity conditions that he verified for every discrete memoryless
MAC and for the Gaussian MAC, Kosut obtained an $O(n^{-1/2})$ converse
remainder. His result therefore identifies the correct second-order scale,
although its converse coefficients do not generally match the existing
achievability coefficients \cite{MolavianJaziLaneman2015,ScarlettMartinezGuillen2015,yavas21}.

In a complementary maximal-error result, Wei and
Kosut~\cite{WeiKosut2021Wringing} combined wringing with a
meta-converse~\cite{PPV10} to obtain an $O(n^{-1/2})$ converse remainder for
the two-user discrete memoryless MAC. When the product input distribution
maximizing the usual average-error sum capacity is unique, they further
derived a normal approximation upper bound for the sum rate involving the
corresponding dispersion coefficient. Because this bound is centered at the
average-error sum capacity, which can exceed the maximal-error sum capacity,
it does not, in general, characterize the exact maximal-error second-order
region.

Thus, for the Gaussian MAC under the average-error criterion, the
$n^{-1/2}$ scale of the converse remainder was known, but the exact
second-order coefficients at the corner points remained open. The present
paper closes this gap at both corner points of the capacity region by proving a
second-order outer bound that coincides with the achievability results
in~\cite{MolavianJaziLaneman2015,ScarlettMartinezGuillen2015,yavas21}.

\subsection{Main Contribution and Techniques}
The main contribution of this paper is a second-order converse showing that
the inner bounds in
\cite{MolavianJaziLaneman2015,ScarlettMartinezGuillen2015,yavas21}
are tight at the two corner points of the Gaussian-MAC capacity region under
the average-error criterion. In particular, our converse recovers the full bivariate Gaussian boundary, including both the covariance between the individual-rate and sum-rate fluctuations and the additional sum-rate dispersion arising from the random inner product of the two independently selected codewords, rather than merely identifying the $1/\sqrt n$ scale of the remainder. The proof combines several
classical converse tools with new quantitative reductions tailored to the
Gaussian MAC. Its central challenge is to convert rate proximity to a corner
point into geometric regularity of the codebooks while preserving their
Cartesian-product structure and, consequently, the independence of the two
 codewords.

After a two-coordinate Yaglom reduction to exact spherical codebooks, we
apply a two-active-constraint Verd\'u--Han bound~\cite{VH94,Han98}.
Row-column trimming then produces a {\em rectangular} subcode with uniform
reliability properties while preserving codeword independence. The
Polyanskiy--Verd\'u good-code output  distribution theorem~\cite{PolVerdu14} and Gaussian maximum entropy control the {\em user-1 log-determinant deficit}, \begin{align}
\Gamma_{1,n}
:=
n\rvC(P_1)
-\frac{1}{2}\log\det\bigl(I_n+K_{1,\rc}^{(n)}\bigr),
\end{align}
where \(K_{1,\rc}^{(n)}\) is the empirical covariance matrix of the trimmed
user-1 codebook. This quantity measures how far the covariance-dependent
Gaussian information
\(\frac12\log\det(I_n+K_{1,\rc}^{(n)})\) falls below the maximum
\(n\rvC(P_1)\) permitted by the power constraint. We show that this deficit
is only \(O(\sqrt n)\), which imposes useful spectral regularity on the
trimmed codebook. A change-of-measure
argument and a covariance-sensitive Gaussian converse similarly control the
user-2 deficit $\Gamma_{2,n}$. %, despite its possibly vanishing  success probability.

These deficits yield a {\em  diffuse-exceptional} decomposition of the raw
second-moment matrix of the trimmed user-2 codebook. On the high-dimensional
diffuse subspace, the second moment operator is small, and the
Brascamp--Lieb projection inequality~\cite{CarlenCorderoErausquin2009},
together with deconvolution, gives the Gaussian limit of the normalized
codeword inner product—the converse analogue of the fluctuation appearing
in~\cite{yavas21}. The complementary exceptional subspace is
low-dimensional; after a small energy expurgation, Gaussian output distributions with
enlarged variances make the positive parts of its likelihood contributions
negligible on the $\sqrt n$ scale. Judicious choices of hybrid  output distributions and
a conditional characteristic function argument then establish the required
joint Gaussian limit, which completes the converse through the
Verd\'u--Han bound.

\subsection{Notation}   Throughout, $\log$ denotes the natural logarithm and $\exp(x)=e^x$; consequently, all rates and information quantities are measured in nats. For a random element $X$, $P_X$ denotes its probability law, and $D(P\|Q)$ (resp.\ $D(p\|q)$) denotes the relative entropy between probability laws $P$ and $Q$ (densities $p$ and $q$). We write $\calN(\mu,\Sigma)$ for the multivariate Gaussian distribution with mean $\mu$ and covariance matrix $\Sigma$, and $\calN(x;\mu,\Sigma)$ for its density evaluated at $x$. Convergence in distribution and convergence in probability are denoted by $\Rightarrow$ and $\xrightarrow{\rm P}$, respectively.

For vectors $x$ and $y$, $\langle x,y\rangle$ and $\|x\|$ denote the Euclidean inner product and norm. For a matrix $A$, $A^\top$, $\tr(A)$, $\det(A)$, and $\|A\|_{\mathrm{op}}$ denote its transpose, trace, determinant, and operator norm, respectively. We write $I_d$ for the $d\times d$ identity matrix and $A\preceq B$ when $B-A$ is positive semidefinite. For a subspace $\calS$, $\Pi_{\calS}$ denotes the orthogonal projection onto $\calS$. We also write $u^+:=\max\{u,0\}$ for a real number $u$.

As is customary in information theory, a superscript $n$, as in $x^n$, generally indicates a length-$n$ vector and should not be interpreted as a power. For matrices and other quantities indexed by $n$, we sometimes use a parenthesized superscript, as in $K^{(n)}$. We use the standard asymptotic notation $O(\cdot)$ and $o(\cdot)$, with all asymptotics taken as $n\to\infty$. The implicit constants may depend on fixed model parameters, but not on $n$.

\section{Problem Setting, Definitions and Main Result}
In this section, we state the channel model, various definitions, some known results, and our main result.

\subsection{Gaussian MAC  model}
The signal model for the two-user Gaussian MAC is given by 
\begin{align}
Y=  X_1 + X_2 + Z, \label{eqn:channel}
\end{align}
where $X_1$ and $X_2$ represent the inputs to the channel,  $Z\sim\calN(0,1)$ is additive Gaussian noise with zero mean and unit variance, and $Y$ is the output of the channel. Thus, the   channel  from $(X_1, X_2)$ to $Y$ can be written as 
\begin{align}
\rvW(y|x_1, x_2) = \frac{1}{\sqrt{2\pi}}\exp\left(-\frac{1}{2}(y-x_1-x_2)^2 \right).
\end{align}
The channel is used $n$ times in a memoryless manner without feedback. The channel inputs (i.e., the transmitted codewords) $x_1^n = (x_{1,1},\ldots, x_{1,n})$ and $x_2^n = (x_{2,1},\ldots, x_{2,n})$ are required to  satisfy the per-codeword power constraints
\begin{align}
\|x_j^n\|^2\le n P_j,\quad j\in \{1,2\}, \label{eqn:power_constraints}
\end{align}
where  $P_1$ and $P_2$    are arbitrary positive   numbers (powers). We omit deterministic channel gains without loss of generality. Indeed,
for the model
$
Y=g_1U_1+g_2U_2+Z,
$
define $X_j=g_jU_j$. If $\|u_j^n\|_2^2\le n\bar P_j$, then
$\|x_j^n\|_2^2\le nP_j$, where $P_j=g_j^2\bar P_j$.
Thus, the gains may be absorbed into the effective received powers.

\subsection{Definitions and known first-order regions}
\begin{definition}[Gaussian MAC Code] \label{def:code}
An  {\em $(n,M_{1,n},M_{2,n},P_1, P_2,\eps_n)$-code}  for the Gaussian MAC consists of two message sets $\calM_{1,n}$ and $ \calM_{2,n}$ of sizes $M_{1,n}$ and $M_{2,n}$ respectively, two encoders $f_{1,n} : \calM_{1,n}\to \bbR^n$ and      $f_{2,n} : \calM_{2,n}\to\bbR^n$  and a decoder    $d_n : \bbR^n\to \calM_{1,n} \times \calM_{2,n}$ such that
\begin{align}
  \|f_{1,n}(m_1)\|^2 &\le n P_1    \quad\forall  \, m_1 \in \calM_{1,n}  , \label{eqn:power1} \\
 \|f_{2,n}(m_2)\|^2&\le n P_2    \quad\forall \,   m_2 \in  \calM_{2,n}   \label{eqn:power2}, \\
\Pr\big( (W_1, W_2)\ne (\hatW_1,\hatW_2) \big)  &\le\eps_n , \label{eqn:error_prob1}
\end{align}
where the messages $W_1$ and $W_2$ are independent and uniformly distributed on $\calM_{1,n}$ and $\calM_{2,n}$ respectively, and $(\hatW_1,\hatW_2) :=d_n(Y^n)$ are the decoded messages. 
\end{definition}

\begin{definition}[First-Order Coding Rates]
A pair of  non-negative numbers $(R_1, R_2)$ is {\em $\eps$-achievable} if there exists a sequence of $(n,M_{1,n},M_{2,n},P_1, P_2,\eps_n)$-codes such that  
\begin{align}
\liminf_{n\to\infty}\frac{1}{n}\log M_{j,n} \ge R_j,\quad j = 1,2,\quad\mbox{and}\quad \limsup_{n\to\infty}\eps_n \le \eps.
\end{align}
The {\em $\eps$-capacity region} $\calC(\eps) \subset\bbR_+^2$ is defined as the closure of the set of  $\eps$-achievable rate pairs $(R_1, R_2)$.  The {\em capacity-region}  $\calC$ is defined as % $\calC$ is the intersection of all $\eps$-capacity regions for $0< \eps <1$. 
\begin{equation}
\calC:=\bigcap_{\eps > 0}\calC(\eps)=\lim_{\eps\to 0}\calC(\eps).\label{eqn:cap_region}
\end{equation}
%where the   limit exists because of the monotonicity of $\calC(\eps)$.
\end{definition}

For ease of notation, we define the Gaussian capacity and dispersion functions respectively as 
\begin{align}
\rvC(P):= \frac{1}{2}\log(1+P)\quad\mbox{and}\quad \rvV(P):=\frac{P(P+1)^2}{2(P+2)^2}.
\end{align}
The following notations will also simplify the following exposition:
\begin{align}
P_\Sigma:=P_1+P_2,\quad q_1:=P_1+1,\quad q_\Sigma:= P_1+P_2+1.
\end{align}
Let $\calC_{\mathrm{CW}}$, the Cover--Wyner region, denote the set of rate pairs $(R_1, R_2)$ satisfying~\eqref{eqn:cap_region_intro}.
% \begin{align}
% R_1 \le \rvC(P_1), \qquad
% R_2 \le \rvC(P_2) ,\qquad
% R_1+R_2 \le \rvC(P_{\Sigma}) .
% \end{align}
  Cover~\cite{Cover1975SomeAdvances} and Wyner~\cite{wyner74} characterized the capacity region.
\begin{theorem}[Capacity Region] \label{thm:first-order-weak}
The capacity region of the Gaussian MAC  is
\begin{equation}
\calC = \calC_{\mathrm{CW}}.
\end{equation}
\end{theorem}
Since the capacity region in \eqref{eqn:cap_region} is defined with the average error probability $\eps$ tending to zero, the result in Theorem~\ref{thm:first-order-weak} combines both an achievability and a {\em weak converse}. Fong and Tan~\cite{fongtan16} established the {\em strong converse}, thus showing the following result. 
\begin{theorem}[$\eps$-Capacity Region] \label{thm:first-order-strong}
The $\eps$-capacity region of the Gaussian MAC  is
\begin{equation}
\calC(\eps) = \calC_{\mathrm{CW}}\qquad\forall\ 0\le \eps<1.
\end{equation}
\end{theorem}

Next, we state the definition of  {\em local} second-order coding rates~\cite{Tan14}.   Here    $(R_1^*, R_2^*)$ is a pair of rates   on the boundary of $\calC(\eps)$ or, equivalently, $\calC$. 

\begin{definition}[Second-Order Coding Rates] \label{def:second}
A pair of numbers $(L_1, L_2)$ is {\em $(\eps,R_1^*,R_2^*)$-achievable} if there exists a sequence of $(n,M_{1,n},M_{2,n},P_1, P_2,\eps_n)$-codes such that  
\begin{align}
\liminf_{n\to\infty}\frac{1}{\sqrt{n}}(\log M_{j,n} - nR_j^*) \ge L_j,\quad j \in\{ 1,2\},\quad\mbox{and}\quad \limsup_{n\to\infty}\eps_n \le \eps. \label{eqn:second_def}
\end{align}
The {\em $(\eps,R_1^*,R_2^*)$-optimal second-order coding rate region} $\calL(\eps;R_1^*,R_2^*) \subset\bbR^2$ is  the closure of the  set of  $(\eps,R_1^*,R_2^*)$-achievable rate pairs  $(L_1, L_2)$. 
\end{definition}
%Stated differently, if $(L_1, L_2)$ is   $(\eps,R_1^*,R_2^*)$-achievable, then there are codes  whose error probabilities are asymptotically no larger than $\eps$, and whose sizes $(M_{1,n},M_{2,n})$ satisfy the asymptotic relation in~\eqref{eqn:roughly}.   

Our goal is to characterize the second-order coding region at the two
corner points of $\calC_{\mathrm{CW}}$.
%and on the individual-rate faces of
%$\calC_{\mathrm{CW}}$. We leave the treatment of the relative interior of the
%dominant sum-rate face to a subsequent paper.

\subsection{Existing second-order inner bound}
We focus on the {\em first-corner point}
$(R_1^*,R_2^*)=(\rvC(P_1),\rvC(P_\Sigma)-\rvC(P_1) )$
in this work.
The corresponding result at the second corner follows by interchanging
the two users. On the relative interiors of the vertical and horizontal
faces, only one individual-rate constraint is active, so their
second-order characterizations follow from the scalar point-to-point
Gaussian result~\cite{Hayashi09,PPV10}. We do not address the relative interior of the dominant
sum-rate face, whose matching second-order converse requires a separate
argument. See Section~\ref{sec:concl} for discussions on the difficulties of the second-order analysis at the relative interior of the sum-rate face.

To state the second-order inner bound at the first corner, define the
{\em first-corner  dispersion matrix}
\begin{equation}
\bV_{\mathrm{fc}}
:=
\begin{bmatrix}
\rvV(P_1) & V_{1,\Sigma}\\
V_{1,\Sigma} & V_{\Sigma,\Sigma}
\end{bmatrix}. \label{eqn:defVfc}
\end{equation} 
where 
\begin{equation}
V_{1,\Sigma} := \frac{P_1(2+P_\Sigma)}{2(1+P_1)(1+P_\Sigma)}\quad\mbox{and}\quad V_{\Sigma,\Sigma}:=\rvV(P_\Sigma)+\frac{P_1P_2}{q_\Sigma^2}.
\end{equation}
In addition, for a $2\times 2$ positive-semidefinite matrix $\bV$, let $(Z_1, Z_2)\sim\calN(\bzero,\bV)$ and define $\Psi(z_1, z_2; \bV):=\Pr(Z_1\le z_1, Z_2 \le z_2)$. Following works on second-order asymptotics in multi-user information theory (e.g., Scarlett and Tan~\cite{ST2015}), we define the set-valued inverse 
\begin{equation}
\Psi^{-1}(\bV,\eps):=\left\{ (z_1, z_2) \in\bbR^2: \Psi(z_1,z_2;\bV)\ge 1-\eps\right\}.
\end{equation}
Define the {\em  information  first-corner second-order  region} as
\begin{align}
    \calL_{\mathrm{fc}}(\eps;P_1, P_2):=\left\{(L_1, L_2)\in\bbR^2: -\begin{bmatrix}
L_1\\ L_1+L_2
\end{bmatrix} \in\Psi^{-1}(\bV_{\mathrm{fc}},\eps) \right\}. \label{eqn:inf_fc}
\end{align}
The following second-order inner bound was established  by MolavianJazi
and Laneman~\cite{MolavianJaziLaneman2015} using a functional central limit theorem (CLT) and by Scarlett, Martinez,
and Guill\'en i F\`abregas~\cite{ScarlettMartinezGuillen2015} by analyzing constant composition codes. It also
follows by specializing the finite blocklength achievability theorem of
Yavas, Kostina, and Effros~\cite{yavas21} to the two active constraints
at the first corner.
\begin{theorem}[First-corner second-order inner bound~\cite{MolavianJaziLaneman2015, ScarlettMartinezGuillen2015, yavas21}] \label{thm:yke}
For every $\eps\in (0,1)$, %The optimal second-order coding rate region at the first-corner point can be inner bounded as 
\begin{equation}
\calL(\eps; \rvC(P_1),  \rvC(P_\Sigma)-\rvC(P_1))\supseteq \calL_{\mathrm{fc}}(\eps;P_1, P_2).
\end{equation}
\end{theorem}
To explain the structure of $\bV_{\mathrm{fc}}$, we briefly recall the
achievability proof of Yavas, Kostina, and Effros \cite{yavas21}. Their  achievability employs an independent
product-of-spheres random-coding ensemble: for each user $j \in \{1,2\}$, codewords are
drawn independently and uniformly from the sphere of radius $\sqrt{nP_j}$,
and the receiver applies maximum-likelihood decoding analyzed through the MAC
random coding union bound. A distinctive feature of this ensemble is the
random cross-inner-product between the two transmitted codewords. In
particular, the centered sum-rate information density contains the term
$\langle X_1^n,X_2^n\rangle/q_\Sigma$. Spherical symmetry yields
\begin{equation}
 T_n:=\frac{\langle X_1^n,X_2^n\rangle}{\sqrt n}
 \Rightarrow \calN(0,P_1P_2). \label{eqn:Tn}
\end{equation}
Consequently, the cross-inner product contributes the additional dispersion
${P_1P_2}/{q_\Sigma^2}$
to the sum-rate coordinate, giving
$
 V_{\Sigma,\Sigma}
 =\rvV(P_\Sigma)+{P_1P_2}/{q_\Sigma^2}.
$
This contribution is distinct from the off-diagonal entry
$V_{1,\Sigma}$, which is the covariance between the user-1 and sum-rate information-density
fluctuations induced by their common Gaussian noise. Conditional on the
cross-inner-product, a suitably modified information-density vector is
amenable to a multivariate Berry--Esseen approximation. Combining this
conditional approximation with the asymptotic Gaussianity of $T_n$ yields
the joint Gaussian limit with covariance matrix $\bV_{\mathrm{fc}}$, whose
bivariate Gaussian cumulative distribution function determines the
first-corner second-order inner bound.

\subsection{Main result}
Our main result is a converse matching
Theorem~\ref{thm:yke}, thereby giving a complete
characterization of the second-order coding region at the first corner.

\begin{theorem}[First-corner second-order outer bound]\label{thm:second-order_outer}
For every $\eps\in(0,1)$, %The optimal second-order coding rate region at the first-corner point can be outer bounded as 
\begin{equation}
\calL(\eps; \rvC(P_1),  \rvC(P_\Sigma)-\rvC(P_1))\subseteq \calL_{\mathrm{fc}}(\eps;P_1, P_2). \label{eqn:outer_bd}
\end{equation}
\end{theorem}

In view of Theorems~\ref{thm:yke} and~\ref{thm:second-order_outer}, the following corollary holds. 
\begin{corollary}[Optimal second-order  region at first corner]
For every $\eps\in(0,1)$,%The optimal second-order coding rate region at the first-corner point 
\begin{equation}
\calL(\eps; \rvC(P_1),  \rvC(P_\Sigma)-\rvC(P_1))= \calL_{\mathrm{fc}}(\eps;P_1, P_2).
\end{equation}
\end{corollary}
Section \ref{sec:proof4} is devoted to proving Theorem \ref{thm:second-order_outer}.

\section{Proof of Theorem \ref{thm:second-order_outer}} \label{sec:proof4}
The converse proof proceeds in five stages. 

First, a two-coordinate Yaglom
map converts the per-codeword power constraints into exact spherical
constraints without changing the inner products between the users'
codewords. A two-constraint Verd\'u--Han bound~\cite{VH94,Han98} then
identifies the user-$1$ and sum-rate information densities as the two
quantities that govern the first corner. These are detailed in Subsections \ref{sec:yaglom} and \ref{sec:vh_bound}.

Second, we trim the
two message sets to obtain a rectangular product subcode while preserving
its first- and second-order rates, asymptotic error probability, and the
independence of the users' codewords (cf.\ Fig.~\ref{fig:rectangular-trimming}). The Polyanskiy--Verd\'u
good-code output theorem~\cite{PolVerdu14} shows that the distribution of
the retained user-$1$ codeword plus Gaussian noise is within
$O(\sqrt n)$ relative entropy of the capacity-achieving Gaussian output
distribution. A change-of-measure argument then uses this Gaussian
distribution as a comparison channel for user~2 and yields quantitative
constraints on the empirical covariance matrices of both trimmed
codebooks.  These are done in Subsections \ref{sec:trim} and \ref{sec:user2}.

Third,  after establishing a novel convariance-specific converse with explicit dependence on the  success probability in Subsection~\ref{sec:cov_converse}, we apply a spectral decomposition specifically to the
trimmed user-$2$ codebook in Subsection~\ref{sec:split}  (cf.\ Fig.~\ref{fig:diffuse-exceptional}). If $X_{2,\rc}^n$ is uniformly
distributed over that codebook, let
\begin{equation}
S_{2,\rc}^{(n)}
:=\mathbb E \big[
X_{2,\rc}^n(X_{2,\rc}^n)^\top\big]
\end{equation}
be its second-moment matrix.
We let $\mathcal E_n$ (where $\calE$ stands for ``exceptional'') be the span of the eigenvectors of
$S_{2,\rc}^{(n)}$ whose eigenvalues exceed a prescribed threshold,
and set $\mathcal D_n:=\mathcal E_n^\perp$ (where $\calD$ stands for ``diffuse''). 
% Thus, for every unit vector
% $v\in\mathcal D_n$, the average squared projection
% $\mathbb E[\langle v,X_{2, \rc}^n\rangle^2]$ is small, whereas
% $\mathcal E_n$ contains the relatively few directions along which the
% user-$2$ codewords may be strongly concentrated. 
Preceding covariance bounds ensure that $\mathcal E_n$ has sufficiently small dimension and that its projected codeword norms can be controlled after removing a vanishing fraction of messages. This is detailed in Subsection~\ref{sec:energy_expurg}. 

Fourth, in Subsection \ref{sec:bl}, we use  the Brascamp--Lieb projection
inequality~\cite{CarlenCorderoErausquin2009,LiuCourtadeCuffVerdu2016}, together with the user-$1$ output divergence bound, to prove that 
\begin{equation}
\frac{\left\langle
\Pi_{\mathcal D_n}X_1^n,\Pi_{\mathcal D_n}X_2^n
\right\rangle}{\sqrt n}
\Rightarrow \mathcal N(0,P_1P_2).
\end{equation}
This is the diffuse-subspace analogue of the overlap variable $T_n$
appearing in the Yavas--Kostina--Effros achievability argument;  cf.\ \eqref{eqn:Tn}. In Subsection \ref{sec:broad}, we show that no convergence in distribution to a Gaussian  is required
for the corresponding overlap on $\mathcal E_n$. Instead, in Subsection~\ref{sec:info_dens}, we choose  output distributions in the
Verd\'u--Han likelihood ratios 
that have the usual capacity-matching Gaussian variance on
$\mathcal D_n$ but an enlarged Gaussian variance on $\mathcal E_n$ (see Fig.~\ref{fig:hybrid-output-laws}).
This enlarged variance makes the entire contribution from
$\mathcal E_n$, including its possibly non-Gaussian cross-inner-product,
negligible at the $\sqrt n$ scale. 

Finally, in Subsection~\ref{sec:char_functions},  a characteristic function
calculation combines the Gaussian diffuse overlap with the Gaussian-noise
fluctuations. The completion of the proof in Subsection~\ref{sec:completion} involves substitution into the Verd\'u--Han bound and yields the
second-order outer bound in
Theorem~\ref{thm:second-order_outer}.

Some explicit parameters used throughout are stated here for notational convenience. 
\begin{align}
t:=\frac{1+\eps}{2},\quad \rho_0:=1-t=\frac{1-\eps}{2},\quad \gamma_n:=n^{1/4},\quad \tau_n:=n^{3/8},\quad b_n :=n^{7/8}. \label{eqn:params2}
\end{align}

\subsection{Spherical reduction by a two-coordinate Yaglom map} \label{sec:yaglom}
The subsequent arguments will be simpler when every codeword uses all the
available power. The following proposition shows that imposing this
exact-shell condition changes neither the error probability nor the
first- and second-order rate coefficients. We use two distinct padding
coordinates, one for each user. Their orthogonality is essential because
it preserves the inner product between the users' codewords.
This is a two-user adaptation of the  Yaglom map trick \cite[Sec.~2]{SoleBelfiore2013}; see also
\cite[Ch.~9, Thm.~6]{ConwaySloane1999}.

\begin{proposition}[Spherical reduction] \label{prop:spherical}
Let $n = \tiln+2$. 
Every length-$\tiln$ Gaussian MAC code satisfying the
per-codeword power constraints (cf.\ \eqref{eqn:power1} and \eqref{eqn:power2}) can be converted into a length-$n$ code with the same message sizes and the same conditional error probabilities such that
\begin{align}
\|\barx_j^n(m_j)\|^2 = nP_j,\qquad j \in \{1,2\}
\end{align}
and 
\begin{align}
\langle\barx_1^n(m_1), \barx_2^n(m_2) \rangle=\langle x_1^{\tiln}(m_1), x_2^{\tiln}(m_2) \rangle. \label{eqn:inner_prod_inv}
\end{align}
The change from $\tiln$ to $n=\tiln+2$ changes the first-order
centering from $\tiln R_j^*$ to $nR_j^*$ by only $2R_j^*=O(1)$ and leaves
the coefficient of $\sqrt n$ unchanged.
%The replacement $\tiln\mapsto n$ changes each first-order term (in the asymptotic expansions of $\log M_{j,n}$ in \eqref{eqn:second_def}) by $O(1)$ and leaves every $\sqrt{n}$-coefficient (in the same asymptotic expansions) unchanged. 
\end{proposition}
\begin{proof}
Begin with a length-$\tiln$ code satisfying~\eqref{eqn:power1} and~\eqref{eqn:power2} with $\tiln$ in place of $n$. Let the codebooks be $\{  x_1^{\tiln}(m_1): m_1 \in\calM_{1,\tiln} \}$ and $\{  x_2^{\tiln}(m_2)  :m_2 \in\calM_{2,\tiln} \}$. Define  the length-$n$ vectors
\begin{align}
 \barx_1^n(m_1) & := \big( x_1^{\tiln}(m_1), \sqrt{nP_1-\|x_1^{\tiln}(m_1)\|^2},0\big),\\
  \barx_2^n(m_2) & := \big( x_2^{\tiln}(m_2), 0,\sqrt{nP_2-\|x_2^{\tiln}(m_2)\|^2} \big).
\end{align}
The square roots are well defined as $\|x_j^{\tiln}(m_j)\|^2 \le \tiln P_j$ for $j = 1,2$ and $n=\tiln+2$. The newly defined codewords lie exactly on the sphere, i.e., $\| \barx_j^n(m_j)\|^2 =nP_j$ for $j = 1,2$. The two appended coordinates are orthogonal between users so~\eqref{eqn:inner_prod_inv} holds.

The new decoder ignores the final two observations and applies the original decoder to the first $\tiln$ coordinates. Hence every conditional error probability is unchanged. Moreover, for any fixed first-order coefficient $R_j^*$, 
\begin{align}
\tiln R_j^* = nR_j^*-2R_j^*= nR_j^*+O(1)  \quad\mbox{and}\quad\sqrt{\tiln} = \sqrt{n}+o(1),
\end{align}
so the first- and second-order coefficients are unchanged.  
\end{proof}
Henceforth, we relabel the extended codewords $\barx_j^n$ as $x_j^n$.
Proposition~\ref{prop:spherical} therefore allows us to assume that 
\begin{equation}
\|x_j^n (m_j)\|^2 =n P_j\quad \mbox{for every  } j \in \{1,2\},~m_j\in\calM_{j,n}. \label{eqn:spherical}
\end{equation}
Every subsequent expurgation removes complete messages without
modifying the retained codewords and therefore preserves~\eqref{eqn:spherical}.
\subsection{A two-constraint Verd\'u--Han bound}\label{sec:vh_bound}
At the first corner, only the user-1 individual constraint and the sum-rate constraint are active at
first order. The following lemma is the converse entry point: it upper-bounds correct decoding by
the joint upper tail of the corresponding two information densities. The inactive user-2 constraint
may be omitted because doing so only enlarges that event. This is a two-constraint specialization of Han's information-spectrum
converse for the MAC~\cite{Han98}; its point-to-point antecedent is the celebrated
Verd\'u--Han converse~\cite{VH94}.

To state the lemma succinctly,   let $\rvW^n$ denote the $n$-fold memoryless extension of the Gaussian
MAC transition law.  Given any   output distributions $Q_{Y^n|X_2^n}$ and $Q_{Y^n}$, define the information densities
\begin{align}
\iota_{1,n} (x_1^n,x_2^n,y^n) &:= \log \frac{\rvW^n(y^n|x_1^n,x_2^n)}{Q_{Y^n|X_2^n} (y^n|x_2^n)} , \label{eqn:id1}\\
\iota_{\Sigma,n} (x_1^n,x_2^n,y^n) &:= \log \frac{\rvW^n(y^n|x_1^n,x_2^n)}{Q_{Y^n} (y^n )} . \label{eqn:id_sum}
\end{align}
The lemma itself does not require power constraints.
\begin{lemma}[Two-active-constraint Verd\'u--Han converse for the MAC] \label{lem:vh} For every $\gamma>0$ and every pair of   output distributions $(Q_{Y^n|X_2^n},Q_{Y^n})$, any $(n,M_{1,n},M_{2,n},\eps_n)$-code satisfies 
\begin{align}
1-\eps_n\le\Pr \left(\begin{bmatrix}
\iota_{1,n} (X_1^n,X_2^n,Y^n) \\ \iota_{\Sigma,n} (X_1^n,X_2^n,Y^n)
\end{bmatrix} \ge \begin{bmatrix}
\log M_{1,n} -\gamma\\ \log (M_{1,n} M_{2,n}) -\gamma
\end{bmatrix} \right) + 2e^{-\gamma},
\end{align}
where the vector inequality is interpreted componentwise and the
probability is evaluated under the distribution induced by the
independent uniform messages, the encoders, and the channel.
\end{lemma}
This is the specialization of Han's MAC information spectrum converse~\cite{Han98} to the active sets
$\{1\}$ and $\{1,2\}$, so we omit the proof.  Interested readers may refer to Han's book \cite[Theorem 7.10.2]{Han10}. As in many second-order analyses, the success of the converse hinges on the judicious choice of the   output distributions used to define the information densities in~\eqref{eqn:id1} and~\eqref{eqn:id_sum}.

\subsection{Rectangular trimming and user-1 output regularity} \label{sec:trim}
The original code has only an average-error guarantee,
whereas the later argument needs two uniform properties: every retained user-1 message must
induce a reliable simulated point-to-point decoder, and every retained user-2 message must have a
positive success probability on the effective additive-noise channel. Selecting rows and then columns
supplies these properties while retaining a Cartesian product of message sets. Consequently the
two message-uniform codewords remain exactly independent, as required by the Brascamp--Lieb
projection step, and the constant-fraction trims change each logarithmic code size by only $O(1)$.

Throughout this subsection, fix an
$(\eps,\rvC(P_1),\rvC(P_\Sigma)-\rvC(P_1) )$-second-order achievable pair \((L_1,L_2)\), and a corresponding sequence
of codes. Thus, %for some sequence \(\eta_n\downarrow0\),
\begin{align}
\log M_{1,n}
&\ge n\rvC(P_1)+(L_1-o(1))\sqrt n, \label{eqn:standing-rate1}\\
\log M_{2,n}
&\ge n\bigl(\rvC(P_\Sigma)-\rvC(P_1)\bigr)
 +(L_2-o(1))\sqrt n, \label{eqn:standing-rate2}
\end{align}
and $\limsup_{n\to\infty}\eps_n\le\eps$.
Following the spherical reduction in Proposition~\ref{prop:spherical},
we also assume   that \eqref{eqn:spherical} holds. 
\begin{proposition}[Rectangular trimming and user-1 output regularity] \label{prop:trim1}
For the spherical first-corner code sequence fixed above and all
sufficiently large \(n\),   there exists message sets\footnote{The superscript $^{\rc}$ stands for ``row/column'' trimmed.} $\calM_{j,n}^{\rc}\subseteq\calM_{j,n}$ for $j =1,2$ such that 
\begin{itemize}
\item[(i)] $|\calM_{j,n}^{\rc}|/ M_{j,n}$ is bounded away from zero (by a positive constant) and  hence,
\begin{align}
\log |\calM_{j,n}^{\rc}|= \log M_{j,n} +O(1) ; \label{eqn:large_subcode}
\end{align}
\item[(ii)] the average error of the rectangular subcode is at most $\eps_n$;
\item[(iii)] every retained user-1 message has row-average error at most $t=\frac{1}{2}( 1+\eps)$   over the original full user-2
message set, and every retained user-2 message has column-average error at most $t$ over~$\calM_{1,n}^{\rc}$. 
\item[(iv)] let \(W_j^{\rc}\) be independent and uniform on
\(\calM_{j,n}^{\rc}\), and define
\(X_{j,\rc}^n:=x_j^n(W_j^{\rc})\). Let
\(Z^n\sim\calN(0,I_n)\) be independent of these codewords and recall that 
\(q_1:=1+P_1\). Then
\begin{equation}
\Delta_{1,n}
:=
D\!\left(
P_{X_{1,\rc}^n+Z^n}
\,\middle\|\,
\calN(0,q_1I_n)
\right)
=O(\sqrt n).\label{eqn:Delta1}
\end{equation}
Writing
\(K_{1,\rc}^{(n)}:=\operatorname{Cov}(X_{1,\rc}^n)\), we have the following bounds on the user-1 log-determinant deficit:
\begin{equation}
0\le\Gamma_{1,n}
:=
n\rvC(P_1)
-\frac12\log\det(I_n+K_{1,\rc}^{(n)})
\le\Delta_{1,n}
=O(\sqrt n).
\end{equation}

%let $W_j^{\rc}$ be independent and uniform on $\calM_{j,n}^{\rc}$ and set $X_{j,\rc}^n:= x_j^n(W_j^{\rc})$. Then 
%\begin{align}
%\Delta_{1,n}:=D \big(P_{X_{j,\rc}^n+Z^n}\| \calN(0,q_1 I_n) \big)= O(\sqrt{n}). \label{eqn:Delta1}
%\end{align}
%Writing $K_{1,\rc}^{(n)}=\mathrm{Cov}(X_{j,\rc}^n)$, we also can bound the log-determinant deficit of user-1 $\Gamma_{1,n}$ as follows:
%\begin{align}
%0\le\Gamma_{1,n}:=n\rvC(P_1)-\frac{1}{2}\log\det(I_n+K_{1,\rc}^{(n)})\le\Delta_{1,n}=O(\sqrt{n}). \label{eqn:logdet1}
%\end{align}
\end{itemize}
The retained code with message sets $\calM_{1,n}^{\rc}$ and $\calM_{2,n}^{\rc}$  is a product code, so its two  codewords are independent.
\end{proposition}

% Manuscript-ready TikZ figure.
%
% REQUIRED IN THE MAIN PAPER PREAMBLE (before \begin{document}):
%   \usepackage{graphicx}
%   \usepackage{tikz}
%   \usetikzlibrary{arrows.meta,patterns}
%
% Then place
%   \input{mac_rectangular_trimming_figure}
% where the figure should be declared.  Change figure* to figure if a
% single-column float is preferred.

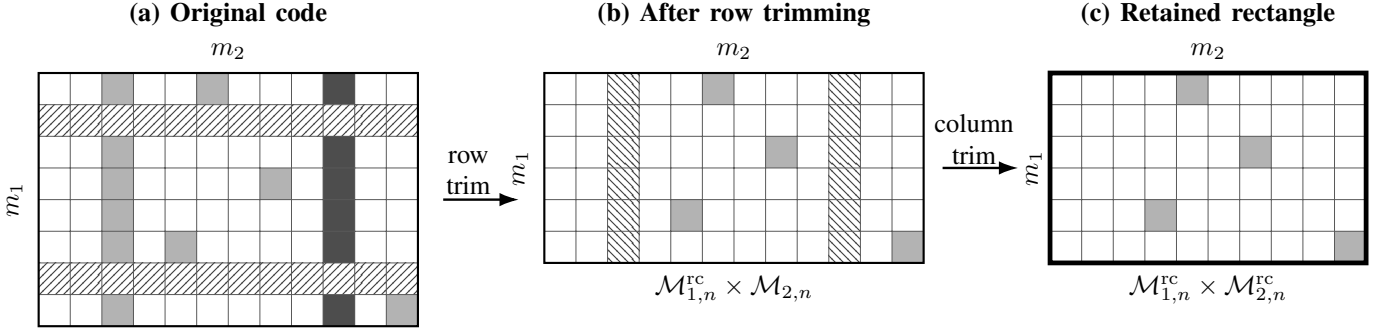
\begin{figure}[t]
\centering
\begingroup

% All style names and helper commands are local to this figure environment.
\tikzset{
  process arrow/.style={-{Latex[length=2.6mm,width=1.8mm]},
                        line width=0.9pt},
  panel title/.style={font=\small\bfseries,align=center},
  axis label/.style={font=\small},
  grid line/.style={draw=black!65,line width=0.28pt},
  outer border/.style={draw=black,line width=0.8pt},
  final border/.style={draw=black,line width=1.8pt},
  high cell/.style={fill=black!70},
  medium cell/.style={fill=black!30},
  row deletion/.style={preaction={fill=white},pattern=north east lines,
                       pattern color=black!75},
  column deletion/.style={preaction={fill=white},pattern=north west lines,
                          pattern color=black!75},
}

\newcommand{\TrimDrawMatrix}[2]{%
  \draw[grid line,step=1] (0,0) grid (#1,-#2);
  \draw[outer border] (0,0) rectangle (#1,-#2);
}
\newcommand{\TrimMediumCell}[2]{%
  \path[medium cell] (#1,-#2) rectangle ++(1,-1);
}
\newcommand{\TrimHighCell}[2]{%
  \path[high cell] (#1,-#2) rectangle ++(1,-1);
}

% The drawing is scaled to the available text width; the caption is not.
\resizebox{\textwidth}{!}{%
\begin{tikzpicture}[x=0.42cm,y=0.42cm,font=\sffamily]

% (a) Original conditional-error matrix.
\begin{scope}[shift={(0,0)}]
  \node[panel title] at (6,1.85) {(a) Original code};
  \foreach \c in {0,...,11}{
    \TrimHighCell{\c}{1}
    \TrimHighCell{\c}{6}
  }
  \foreach \r in {0,2,3,4,5,7}{
    \TrimMediumCell{2}{\r}
    \TrimHighCell{9}{\r}
  }
  \TrimMediumCell{5}{0}
  \TrimMediumCell{7}{3}
  \TrimMediumCell{4}{5}
  \TrimMediumCell{11}{7}
  \path[row deletion] (0,-1) rectangle (12,-2);
  \path[row deletion] (0,-6) rectangle (12,-7);
  \TrimDrawMatrix{12}{8}
  \node[axis label] at (6,0.70) {$m_2$};
  \node[axis label,rotate=90] at (-0.75,-4) {$m_1$};
\end{scope}

\draw[process arrow] (12.8,-4)--(15.2,-4);
\node[font=\small,align=center] at (13.6,-3.15) {row\\trim};

% (b) Retained user-1 rows and all original user-2 columns.
\begin{scope}[shift={(16,0)}]
  \node[panel title] at (6,1.85) {(b) After row trimming};
  \foreach \r in {0,...,5}{
    \TrimMediumCell{2}{\r}
    \TrimHighCell{9}{\r}
  }
  \TrimMediumCell{5}{0}
  \TrimMediumCell{7}{2}
  \TrimMediumCell{4}{4}
  \TrimMediumCell{11}{5}
  \path[column deletion] (2,0) rectangle (3,-6);
  \path[column deletion] (9,0) rectangle (10,-6);
  \TrimDrawMatrix{12}{6}
  \node[axis label] at (6,0.70) {$m_2$};
  \node[axis label,rotate=90] at (-0.75,-3) {$m_1$};
  \node[font=\small] at (6,-6.85)
    {$\mathcal M_{1,n}^{\mathrm{rc}}\times\mathcal M_{2,n}$};
\end{scope}

\draw[process arrow] (28.6,-3)--(31.0,-3);
\node[font=\small,align=center] at (29.6,-2.1) {column\\trim};

% (c) Final Cartesian product of retained message sets.
\begin{scope}[shift={(32,0)}]
  \node[panel title] at (5,1.85) {(c) Retained rectangle};
  \TrimMediumCell{4}{0}
  \TrimMediumCell{6}{2}
  \TrimMediumCell{3}{4}
  \TrimMediumCell{9}{5}
  \TrimDrawMatrix{10}{6}
  \draw[final border] (0,0) rectangle (10,-6);
  \node[axis label] at (5,0.70) {$m_2$};
  \node[axis label,rotate=90] at (-0.5,-3) {$m_1$};
  \node[font=\small] at (5,-6.85)
    {$\mathcal M_{1,n}^{\mathrm{rc}}
      \times\mathcal M_{2,n}^{\mathrm{rc}}$};
\end{scope}

\end{tikzpicture}%
}

\endgroup
\caption{Rectangular trimming of the conditional-error matrix $[  p_n(m_1, m_2)]$,
whose rows are indexed by user-1 messages and whose columns are indexed by
user-2 messages. First, each row is averaged over the original user-2 message
set, and rows with average conditional error exceeding $t$ are trimmed per  \eqref{eqn:M1rc}
(/-shaded  boxes). Column averages are then recomputed using only the
retained rows, and columns whose averages exceed $t$ are trimmed  per \eqref{eqn:M2rc} (\textbackslash-shaded
boxes). The remaining message set is the Cartesian product
$\mathcal M_{1,n}^{\mathrm{rc}}\times
\mathcal M_{2,n}^{\mathrm{rc}}$, so independently uniform retained messages
remain independent. Gray cell intensities are schematic and indicate
different conditional-error levels.}
\label{fig:rectangular-trimming}
\end{figure}

The construction of $\calM_{j,n}^\rc$ for $j = 1,2$ is illustrated in Fig.~\ref{fig:rectangular-trimming}. The proof of item (iv) of Proposition~\ref{prop:trim1} requires the use of the good-code output theorem of Polyanskiy--Verd\'u~\cite{PolVerdu14} so we record it here for convenience. Their deterministic encoder
formulation permits stochastic decoding.
\begin{lemma}[Good-code output distribution theorem~\cite{PolVerdu14}] \label{lem:good_code}
For any fixed $P>0$ and $t\in (0,1)$, there is a finite $a(P,t)$ that any deterministic-encoder AWGN code with $\|x^n(m)\|^2\le nP$ and maximal error at most $t$ satisfies, for the message uniform input law $P_{X^n}(\cdot ) = \frac{1}{M}\sum_{m=1}^M \delta_{x^n(m)} (\cdot)$, 
\begin{align}
D\big(P_{X^n+Z^n} \|\calN(0,(P+1)I_n) \big)\le n\rvC(P)-\log M + a(P,t)\sqrt{n}.
\end{align}
\end{lemma}
\begin{proof}[Proof of Proposition~\ref{prop:trim1}]
For a fixed MAC code, define the conditional probabilities of error 
\begin{align}
p_n(m_1, m_2) :=\Pr \Big( (\hatW_1, \hatW_2) \ne (m_1,m_2) \mid  (W_1,W_2) =(m_1,m_2)\Big)\qquad (m_1,m_2)\in \calM_{1,n}\times \calM_{2,n}.
\end{align}
Fix $t$ as in \eqref{eqn:params2}. Then for all $n$ large enough, $\eps_n<t$.

%Parts (i), (ii), and (iii) arise from the same two-stage trimming construction. The row and column stages below are therefore labelled by every part that they establish.

\paragraph{User-1 row trimming} We now trim the rows of user 1. 
%\emph{Cardinality of the retained row set (parts (i) and (iii)).}
Define row averages  and the retained row set respectively as 
\begin{align}
\bar{\eps}_{1,n}(m_1):=\frac{1}{M_{2,n}}\sum_{m_2 \in \calM_{2,n}}p_n(m_1, m_2)\quad\mbox{and}\quad \calM_{1,n}^{\rc}:= \big\{m_1 \in \calM_{1,n}: \bar{\eps}_{1,n}(m_1)\le t\big\}. \label{eqn:M1rc}
\end{align}
Markov's inequality gives $|\calM_{1,n}^{\rc}|/M_{1,n}\ge 1-\eps_n/t$, which is eventually bounded away from zero and hence \eqref{eqn:large_subcode} holds for message $j=1$.

%\emph{Average error after row trimming (part (ii)).} 
Define the average error after retaining the user-1 rows by
\begin{equation}
\bar{\eps}_{1,n}^\rc :=
\frac{1}{|\calM_{1,n}^{\rc}|M_{2,n}}
\sum_{m_1\in\calM_{1,n}^{\rc}}
\sum_{m_2\in\calM_{2,n}}
p_n(m_1,m_2).
\end{equation}
The original average error is a weighted average of the average errors
over the retained and deleted rows. Every deleted row has average error
greater than \(t>\eps_n\), whereas the original average error is at most
\(\eps_n\). Consequently, the average error over the retained rows cannot
exceed the original average error, and hence
$
\bar{\eps}_{1,n}^\rc\le \eps_n.
$

\paragraph{User-2 column trimming} We now trim the columns indexed by user-2 messages, using averages over the retained user-1 rows with message indices in $\calM_{1,n}^{\rc}$.
Define
\begin{align}
\bar{\eps}_{2,n}(m_2):=\frac{1}{|\calM_{1,n}^{\rc}|}\sum_{m_1\in \calM_{1,n}^{\rc}}p_n(m_1,m_2)\quad\mbox{and}\quad \calM_{2,n}^{\rc}:=\big\{m_2\in \calM_{2,n}: \bar{\eps}_{2,n}(m_2)\le t\big\}. \label{eqn:M2rc}
\end{align}
The average of   $\bar{\eps}_{2,n}(m_2)$ over $m_2\in\calM_{2,n}$  equals $\bar{\eps}_{1,n}^\rc$ and is therefore at most $\eps_n$, and so $| \calM_{2,n}^{\rc}| /M_{2,n}\ge 1-\eps_n/t$, which is eventually bounded away from zero and hence \eqref{eqn:large_subcode} holds for message $j=2$.

\paragraph{Completion of part (ii)} Deleting columns with average above $t>\eps_n$ again cannot increase the average. Denote the
average error of the product subcode with product message set $\calM_{1,n}^\rc\times\calM_{2,n}^\rc$ by 
\begin{equation}
\bar{\eps}_n^\rc :=
\frac{1}
{|\calM_{1,n}^\rc||\calM_{2,n}^\rc|}
\sum_{m_1\in\calM_{1,n}^\rc}
\sum_{m_2\in\calM_{2,n}^\rc}
p_n(m_1,m_2). \label{eqn:mean_errors}
\end{equation}
%\bar{\eps}_n^\rc:=\frac{1}{|\calM_{1,n}^\rc||\calM_{2,n}^\rc| }\sum_{m_1\in \calM_{1,n}^\rc}\sum_{m_2\in \calM_{1,n}^\rc}\eps_n(m_1, m_2). \label{eqn:mean_errors}
Define the decoder for the rectangular subcode by retaining the original
decoder output whenever it belongs to
\(\calM_{1,n}^{\rc}\times\calM_{2,n}^{\rc}\), and otherwise mapping it to
an arbitrary fixed pair in this retained product set. This modification
cannot increase the error probability for any retained message pair.
Consequently, the average error of the resulting rectangular subcode 
$\bar{\eps}_n^\rc\le \bar{\eps}_{1,n}^\rc\le\eps_n$.

%Conditioned on the retained product message set, the messages are independent and uniform on $\calM_{1,n}^\rc\times\calM_{2,n}^\rc$; hence, their average conditional error is the average in \eqref{eqn:mean_errors}. Then $\bar{\eps}_n^\rc\le \bar{\eps}_{1,n}^\rc\le\eps_n$.

\paragraph{Part (iv): user-1 output regularity and log-determinant gap} Let $W_1^\rc$ and $W_2^\rc$ be independent and uniform over $\calM_{1,n}^\rc$ and $\calM_{2,n}^\rc$ respectively and set $X_{j,\rc}^n:= x_j^n(W_j^{\rc})$ for $j = 1,2$. These are the message-uniform codewords of the row/column-trimmed product code. 

Given an AWGN observation $y^n=x_1^n(m_1)+z^n$, define a stochastic user-1 decoder by drawing an independent message $W_2\sim\mathrm{Unif}(\calM_{2,n})$, forming $y^n + x_2^n(W_2)$, running the MAC decoder $d_{m_1, m_2}(\cdot)$, and retaining only the first component. This decoder can be written as 
\begin{align}
\varphi_{m_1}(y^n):=\frac{1}{M_{2,n}}\sum_{m_2 \in \calM_{2,n}}\sum_{m_2' \in \calM_{2,n}}d_{m_1, m_2'}(y^n+x_2^n(m_2)). \label{eqn:varphi_m1}
\end{align}
This family of decoders satisfy $\sum_{m_1 \in \calM_{1,n}^\rc}\varphi_{m_1}(y^n)\le 1$. The family may be completed to a stochastic decoder by assigning any leftover probability to one fixed message in $\calM_{1,n}^\rc$; this cannot decrease any of the success probabilities in \eqref{eqn:varphi_m1}. Thus, for every $m_1 \in \calM_{1,n}^\rc$,
\begin{align}
&\bbE_{Y^n\sim\calN(x_1^n(m_1),I_n)}
       [\varphi_{m_1}(Y^n)]
\nonumber\\
&=
\frac{1}{M_{2,n}}
\sum_{m_2\in\calM_{2,n}}
\Pr\!\left(
\hat W_1=m_1
 \mid
(W_1,W_2)=(m_1,m_2)
\right)
\nonumber\\
&\ge
\frac{1}{M_{2,n}}
\sum_{m_2\in\calM_{2,n}}
\Pr\!\left(
(\hat W_1,\hat W_2)=(m_1,m_2)
\mid
(W_1,W_2)=(m_1,m_2)
\right)
\nonumber\\
&=
1-\bar{\eps}_{1,n}(m_1)
\ge 1-t.
\label{eqn:prob_varphi}
\end{align}
Applying Lemma~\ref{lem:good_code} to \eqref{eqn:prob_varphi} and using \eqref{eqn:standing-rate1} and~\eqref{eqn:large_subcode} gives
\begin{align}
\Delta_{1,n}:=D\big(P_{X_{1,\rc}^n+Z^n}\|\calN(0,q_1 I_n) \big) &\le n \rvC(P_1)-\log |\calM_{1,n}^\rc|+ a(P_1,t)\sqrt{n}\nn\\
&\le   (a(P_1,t)-L_1)\sqrt{n} + o(\sqrt{n})=O(\sqrt{n}),
\end{align}
which yields \eqref{eqn:Delta1}.

Now let $\mu_{1,\rc}^{(n)}:=\bbE[ X_{1,\rc}^n ]$ and $K_{1,\rc}^{(n)}:=\mathrm{Cov}( X_{1,\rc}^n )$. The output has mean $\mu_{1,\rc}^{(n)}$ and covariance $I_n+K_{1,\rc}^{(n)}$. Using the fact that the multivariate Gaussian maximizes differential entropy subject to a covariance constraint, we obtain 
\begin{align}
\Delta_{1,n}\ge D\big(\calN(\mu_{1,\rc}^{(n)},I_n+K_{1,\rc}^{(n)})\|\calN(0,q_1 I_n)\big)
\end{align}
By the two-coordinate Yaglom reduction in Proposition~\ref{prop:spherical}, the code is spherical; the row and column
deletions preserve every retained codeword's norm. Hence, the trace of the covariance plus the  squared of  norm of the mean vector exactly equals the energy, i.e., 
\begin{align}
\tr(K_{1,\rc}^{(n)})+\|\mu_{1,\rc}^{(n)}\|^2=nP_1.
\end{align}
Substitution in the KL divergence formula for two Gaussians yields the following bound on the log-determinant deficit of user-1:
\begin{align}
\Gamma_{1,n}=n\rvC(P_1)-\frac{1}{2}\log\det(I_n+K_{1,\rc}^{(n)})\le\Delta_{1,n}=O(\sqrt{n}). \label{eqn:Gamma1_bd}
\end{align}
Moreover $\Gamma_{1,n}\ge0$. This can be seen by using the concavity of the logarithm and the fact that $\tr(K_{1,\rc}^{(n)})\le nP_1$. Indeed, 
\begin{align}
\frac{1}{n}\log\det(I_n+ K_{1,\rc}^{(n)})\le\log\bigg( 1+\frac{\tr(  K_{1,\rc}^{(n)})}{n} \bigg)\le\log(1+P_1). \label{eqn:nonneg_deficit}
\end{align}
This completes the proof of part (iv) and hence, Proposition~\ref{prop:trim1}.
\end{proof}

\subsection{The effective user-2 channel}\label{sec:user2}
Part (iii) of Proposition~\ref{prop:trim1} retains only those user-2 messages whose average joint decoding error, averaged over the retained user-1 messages, is at most $t$. Define $U^n:=X_{1,\rc}^n+Z^n$, so that the channel output corresponding to user-2 message $m_2$ can be written as $x_2^n(m_2)+U^n$. The next proposition first converts the preceding average-error bound into a lower bound on the probability of correctly decoding $m_2$ over this channel with effective additive-noise $U^n$. It then uses the relative entropy estimate in \eqref{eqn:Delta1} to compare the same decoding rule when $U^n$ is replaced by Gaussian noise with covariance $q_1I_n$. The resulting Gaussian channel success probability may tend to zero, but is at least $\exp[ - O(\sqrt {n })]$.

\begin{proposition} \label{prop:user2}
For each $m_2\in\calM_{2,n}^\rc$, let the decoder's conditional probability of declaring second message be
\begin{align}
\psi_{m_2}(y^n)
:=\sum_{\hatm_1\in\calM_{1,n}}
     d_{\hatm_1,m_2}(y^n).
\end{align}
Then for each such $m_2$,  for all $n$ sufficiently large,
\begin{align}
\alpha_n(m_2):=\bbE_{Y^n \sim P_{U^n+x_2^n(m_2)}}[ \psi_{m_2}(Y^n) ] \ge\rho_0, \label{eqn:def_alpha_n}
\end{align}
where recall that $\rho_0=1-t=\frac{1}{2}(1-\eps)$; cf.~\eqref{eqn:params2}. Furthermore,
\begin{align}
\beta_n(m_2):=\bbE_{Y^n \sim \calN(x_2^n(m_2), q_1I_n) }[ \psi_{m_2}(Y^n) ] \ge\rho_n:=\exp\left(-\frac{\Delta_{1,n}+\log 2}{\rho_0} \right). \label{eqn:def_beta_n}
\end{align}
Consequently, $\log(1/\rho_n)=O(\sqrt{n})$.
\end{proposition}
\begin{proof}
First, note that $\sum_{m_2 \in \calM_{2,n}^\rc}\psi_{m_2}(y^n)\le 1$.  For every $m_2\in\calM_{2,n}^\rc$, 
\begin{align}
\alpha_n(m_2)
%&:=\bbE_{P_{U^n+x_2^n(m_2)}}[\psi_{m_2}(Y^n)] \nonumber\\
&=\frac{1}{|\calM_{1,n}^{\rc}|}
  \sum_{m_1\in\calM_{1,n}^{\rc}}
  \Pr\!\left(
       \hat W_2=m_2
       \,\middle|\,
       (W_1,W_2)=(m_1,m_2)
       \right) \nonumber\\
&\ge
  \frac{1}{|\calM_{1,n}^{\rc}|}
  \sum_{m_1\in\calM_{1,n}^{\rc}}
  \Pr\!\left(
       (\hat W_1,\hat W_2)=(m_1,m_2)
       \,\middle|\,
       (W_1,W_2)=(m_1,m_2)
       \right) \nonumber\\
&=1-\bar{\eps}_{2,n}(m_2)
 \ge 1-t=\rho_0 .
\end{align}
where we used the definition of $\bar{\eps}_{2,n}(m_2)$ in \eqref{eqn:M2rc} and  the final inequality holds (for all $n$ large enough) by the construction of $\calM_{2,n}^\rc$. The additive noise $U^n$ is generally non-Gaussian, but \eqref{eqn:Delta1} says that its law is close in relative
entropy to $\calN(0,q_1I_n)$. We therefore evaluate the same decoder  under this Gaussian replacement. 
% Define
% \begin{align}
% \beta_{n}(m_2):=\bbE_{Y^n \sim \calN(x_2^n(m_2), q_1I_n) }[ \psi_{m_2}(Y^n) ].
% \end{align}
The divergence between the distribution of $U^n$ and $\calN(0,q_1 I_n)$ does not change when we shift the distributions by the same vector $x_2^n(m_2)$. This, together with~\eqref{eqn:Delta1}, give 
\begin{align}
D \big(P_{U^n+ x_2^n(m_2)}\| \calN(x_2^n(m_2),q_1I_n) \big)=\Delta_{1,n}.
\end{align}
Binary data processing through $\psi_{m_2}$ gives 
\begin{align}
\Delta_{1,n} \ge d_{\rmb}\big(\alpha_n(m_2) \| \beta_n(m_2) \big),
\end{align}
where $\beta_n(m_2)$  is defined in \eqref{eqn:def_beta_n}. Furthermore, the binary divergence $d_{\rmb}(\cdot\|\cdot)$ can be lower bounded in terms of the binary entropy $h_\rmb(\cdot)$ as
\begin{align}
    d_{\rmb}\big(\alpha_n(m_2) \| \beta_n(m_2) \big)\ge \alpha_n(m_2)\log\frac{1}{\beta_n(m_2) }-h_{\rmb}\big(\alpha_n(m_2)\big)\ge \alpha_n(m_2)\log\frac{1}{\beta_n(m_2) }-\log 2.
\end{align}
Since $\alpha_n(m_2)\ge\rho_0$, 
\begin{align}
\beta_n(m_2)\ge \exp\left(-\frac{\Delta_{1,n}+\log 2}{\rho_0} \right). \label{eqn:transfer}
\end{align} 
Recalling that $\rho_0>0$ and $\Delta_{1,n}=O(\sqrt{n})$ yields $\log(1/\rho_n)=O(\sqrt{n})$ as desired.
\end{proof}
\subsection{A Gaussian covariance converse with explicit  dependence on the success probability} \label{sec:cov_converse}
Equation \eqref{eqn:transfer} provides only a possibly vanishing lower bound on the Gaussian channel success probability $\rho_n=\exp[ - O(\sqrt {n })]$. Consequently, a covariance-sensitive converse that assumes a fixed, nonvanishing success probability is insufficient. The following lemma is a Gaussian, covariance-sensitive specialization of Augustin's non-asymptotic converse \cite{Augustin1966} with explicit dependence on the per-message success probability. It extends the fixed-success covariance evaluation used by Polyanskiy and Verd\'u~\cite[Theorem~19]{PolVerdu14} to success probabilities that may vanish subexponentially. Since the bound also retains the full empirical covariance matrix of the codebook, it may find applications beyond the present Gaussian MAC problem.

\begin{lemma} \label{lem:cov_conv}
Fix $q>0$, $P\ge0$, $M\in\bbN$ and $0<\rho\le 1$. Let $\{x^n(m):m\in [M]\}\subset\bbR^n$ satisfy $\|x^n(m)\|^2\le nP$. Let $p_m =\calN(x^n(m),qI_n)$. Suppose tests $\phi_m:\bbR^n\to[0,1]$ satisfy 
\begin{align}
\sum_m \phi_m(y^n)\le 1\quad\mbox{and}\quad\bbE_{Y^n\sim p_m}[\phi_m(Y^n)]\ge\rho \quad\mbox{for each } m. \label{eqn:subnorm}
\end{align}
Let 
\begin{align}
\mu^{(n)}:=\frac{1}{M}\sum_m x^n(m)\quad\mbox{and}\quad K^{(n)}:=\frac{1}{M}\sum_m (x^n(m)-\mu^{(n)})(x^n(m)-\mu^{(n)})^\top.
\end{align}
Then
\begin{align}
\log M \le\frac{1}{2}\log\det\bigg(I_n+\frac{K^{(n)}}{q} \bigg)+\log\frac{2}{\rho}+\sqrt{ \bigg(1+\frac{8P}{q}\bigg)n\log\frac{2}{\rho}}. \label{eqn:cov_conv}
\end{align}
\end{lemma}
\begin{proof}
Set $C := qI_n+K^{(n)}$, $J := C^{-1}$ and $B :=I_n-qJ$. Since $K^{(n)}$ is positive semi-definite, $0\preceq B \prec I_n$. Write $\tilx_m^n:=x^n(m)-\mu^{(n)}$ and set $q_{\mathrm{ref}} := \calN(\mu^{(n)}, C)$. Define their log-likelihood ratios 
\begin{align}
    \ell_m(y^n) := \log\frac{p_m(y^n)}{q_{\mathrm{ref}}(y^n)} .
\end{align}
Write the output as $Y^n :=x^n(m)+\sqrt{q}Z^n$ where $Z^n \sim\calN(0,I_n)$. With $D_0:=\frac{1}{2}\log\det(I_n+K^{(n)}/q)$, direct algebra involving Gaussian densities yields
\begin{align}
    \ell_m(Y^n)=D_0+\frac{1}{2}(\tilx_m^n)^\top J \tilx_m^n+\sqrt{q}(\tilx_m^n)^T J Z^n- \frac{1}{2} (Z^n)^\top B Z^n.
\end{align}
Thus $\ell_m(Y^n)$ is the log-likelihood ratio between the transmitted message output with density $p_m$ and the reference density $q_{\mathrm{ref}}$, evaluated at an output $Y^n$ drawn from $p_m$. The mean of this log-likelihood ratio is 
\begin{align}
   \bar{\ell}_m:= \bbE_{Y^n\sim p_m} [\ell_m(Y^n)]=D(p_m\|q_{\mathrm{ref}})= D_0+\frac{1}{2}(\tilx_m^n)^\top J \tilx_m^n-\frac{1}{2}\tr(B).
\end{align}
Define the centered  log-likelihood-ratio $r_m(y^n)=\ell_m(y^n)-\bar{\ell}_m$. Put $v_m=\sqrt{q}J \tilx_m^n$. Under $Y^n\sim p_m$, the centered  log-likelihood-ratio  is the random variable
\begin{align}
    R_m:=r_m(Y^n)=v_m^\top Z^n- \frac{1}{2}\big( (Z^n)^\top B Z^n -\tr(B)\big).
\end{align}
Clearly, $\bbE [R_m]=0$. For every $\theta\ge0$, integration of Gaussians with respect to $Z^n\sim\calN(0,I_n)$ yields the moment generating function of $R_m$
\begin{align}
    \bbE_{p_m}\big[ \!\exp(\theta R_m) \big] = \exp\left(\frac{\theta}{2}\tr(B) \right)\det(I_n+\theta B)^{-1/2}\exp\left(\frac{\theta^2}{2}v_m^\top (I_n+\theta B)^{-1}v_m\right).
\end{align}
Using the bounds $x-\log(1+x)\le x^2/2$ for $x\ge0$ and $(I_n+\theta B)^{-1}\preceq I_n$ we can bound the cumulant generating function of $R_m$ as 
\begin{align}
    \log \bbE_{p_m}\big[\! \exp(\theta R_m) \big] \le\frac{\theta^2}{4}\tr(B^2)+\frac{\theta^2}{2}\|v_m\|^2. \label{eqn:cgf}
\end{align}
The per-codeword power constraint implies $\|\mu^{(n)}\|\le\sqrt{nP} $ and $\|\tilx_m\|^2\le 4 nP$ (triangle inequality). Since $J\preceq q^{-1}I_n$,
\begin{align}
    \|v_m\|^2  = q (\tilx_m^n)^\top J^2 \tilx_m^n\le\frac{4nP}{q} \quad \mbox{and}\quad \tr(B^2)\le n. 
\end{align}
Plugging this bound into \eqref{eqn:cgf} yields
\begin{align}
    \log \bbE_{p_m}\big[ \!\exp(\theta R_m) \big] \le\frac{\theta^2w_n}{2}\quad \mbox{where}\quad w_n:=n\left( \frac{1}{2}+\frac{4P}{q}\right). \label{eqn:cgf_bd}
\end{align}
An application of Chernoff's bound and \eqref{eqn:cgf_bd} gives
\begin{align}
    \bbE_{p_m}\big[\!\mathbbm{1} \{R_m>s\} \big]\le e^{-s^2/(2w_n)}.
\end{align}
Set $s_n=\sqrt{(1+8P/q)n\log(2/\rho)}$. Then $\bbE_{p_m} [\mathbbm{1} \{R_m>s_n\}  ]\le\rho/2$. Define the set $\calA_m:=\{y^n:r_m(y^n)\le s_n\}$. A change of measure yields
\begin{align}
    \bbE_{q_{\mathrm{ref}}}[\phi_m]&\ge\bbE_{q_{\mathrm{ref}}}\big[\phi_m\mathbbm{1}_{\calA_m} \big]\nn\\
    &=\bbE_{p_m}\big[e^{-\ell_m(Y^n)}\phi_m(Y^n)\mathbbm{1}_{\calA_m}\big]\nn\\
    &\ge e^{-\bar{\ell}_m-s_n}\left(\bbE_{p_m}\big[\phi_m(Y^n)-p_m(\calA_m^c) \big] \right)\nn\\
    &\ge\frac{\rho}{2}e^{-\bar{\ell}_m-s_n}.
\end{align}
Summing over $m$ and using the subnormalization of the tests in \eqref{eqn:subnorm}, we obtain
\begin{align}
    1\ge\frac{\rho}{2}e^{-s_n}\sum_m e^{-\bar{\ell}_m}. \label{eqn:use_subnorm}
\end{align}
Jensen's inequality applied to the convex function $x\mapsto e^{-x}$ gives 
\begin{align}
    \sum_m e^{-\bar{\ell}_m}\ge M \exp\left(-\frac{1}{M}\sum_m\bar{\ell}_m\right).
\end{align}
Finally, 
\begin{align}
    \frac{1}{M}\sum_{m} (\tilx_m^n)^\top J \tilx_m^n=\tr(J K^{(n)})=\tr(B) ,\label{eqn:trB}
\end{align}
so $M^{-1}\sum_m\bar{\ell}_m=D_0$. Substitution of  \eqref{eqn:trB} into \eqref{eqn:use_subnorm} and recalling the definition of $s_n$ gives the converse bound in~\eqref{eqn:cov_conv}.
\end{proof}
Now, we apply Lemma~\ref{lem:cov_conv} to our problem. We define
\begin{align}
    \mu_{2,\rc}^{(n)}:=\bbE [X_{2,\rc}^n]\quad\mbox{and}\quad K_{2,\rc}^{(n)}:=\mathrm{Cov}(X_{2,\rc}^n).
\end{align}
\begin{corollary}[User-2 log-determinant deficit] For the trimmed user-2 codebook and covariance matrix $K_{2,\rc}^{(n)}$, we have  the following bound on the user-2 log-determinant deficit
\begin{align}
0\le\Gamma_{2,n}:=n\rvC(P_2/q_1)-\frac{1}{2}\log\det\bigg( I_n+ \frac{K_{2,\rc}^{(n)}}{q_1}\bigg) = O(n^{3/4}). \label{eqn:user2deficit}
\end{align} 
\end{corollary}
\begin{proof}
    Apply Lemma~\ref{lem:cov_conv} to the trimmed user-2 codebook $\{x_2^n(m_2) : m_2 \in\calM_{2,n}^\rc\}$ with $q\leftarrow q_1$, $P\leftarrow P_2$, $\rho\leftarrow\rho_n$, and tests $\{\psi_{m_2}: m_2 \in\calM_{2,n}^\rc\}$. By using the fact that $\log(1/\rho_n)=O(\sqrt{n})$, we obtain from \eqref{eqn:cov_conv} that 
\begin{align}
    \log | \calM_{2,n}^\rc|\le \frac{1}{2}\log \det\bigg( I_n+ \frac{K_{2,\rc}^{(n)}}{q_1}\bigg) +O(n^{3/4}),
\end{align}
since the sum of the remainder terms in \eqref{eqn:cov_conv} is $O(\sqrt{n})+O(n^{3/4})=O(n^{3/4})$. Furthermore, from \eqref{eqn:standing-rate2} and \eqref{eqn:large_subcode},  for any $\eta>0$, $\log | \calM_{2,n}^\rc|=\log M_{2,n}+O(1) \ge n \rvC(P_2/q_1)+(L_2-\eta)\sqrt{n}$ for $n$ large enough. This gives the upper bound in \eqref{eqn:user2deficit}. Moreover, by the same argument used to establish the nonnegativity of $\Gamma_{1,n}$ (see \eqref{eqn:nonneg_deficit}), we can also show that $\Gamma_{2,n}$ is nonnegative. 
\end{proof}

\subsection{Relating the log-determinant deficits to  a spectral split} \label{sec:split}
At this point the coding argument has produced the log-determinant deficits $\Gamma_{1,n} = O(\sqrt{n}) $ and $\Gamma_{2,n} = O(n^{3/4})$, but the remainder of the proof requires geometric rather than rate information:
The use of Brascamp--Lieb later requires small operator norm for the raw user-2 second moment on a large subspace, while the complementary exceptional subspace must have controlled dimension and energy. This motivates the following diffuse-exceptional subspace decomposition result, which is illustrated  in Fig.~\ref{fig:diffuse-exceptional}.

% Add \usepackage{tikz} to the preamble of the main manuscript.
% The display also uses \mathbb{R}; load amssymb if it is not already loaded.
% No additional TikZ libraries are required.
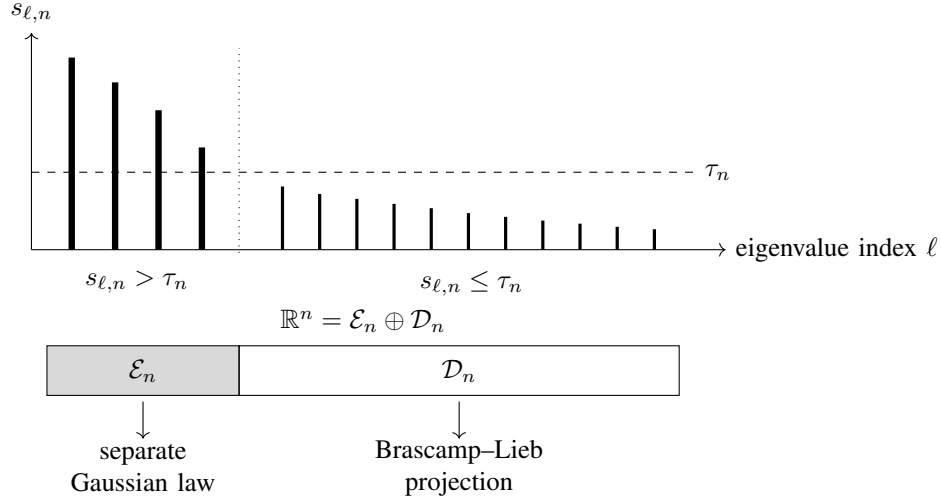
\begin{figure}[t]
    \centering
    \begin{tikzpicture}[x=0.82cm,y=0.82cm,
        every node/.style={font=\small}]

        % Schematic eigenvalue plot.
        \draw[->] (0,0) -- (11.2,0)
            node[right] {eigenvalue index $\ell$};
        \draw[->] (0,0) -- (0,3.5)
            node[above] {$s_{\ell,n}$};
        \draw[dashed] (0,1.25) -- (10.7,1.25)
            node[right] {$\tau_n$};

        % Exceptional eigenvalues: above the cutoff.
        \foreach \x/\h in {0.65/3.10,1.35/2.70,2.05/2.25,2.75/1.65}
            \draw[line width=2.4pt] (\x,0) -- (\x,\h);

        % Diffuse eigenvalues: below the cutoff.
        \foreach \x/\h in {4.05/1.02,4.65/0.90,5.25/0.82,5.85/0.74,
                            6.45/0.67,7.05/0.59,7.65/0.53,8.25/0.47,
                            8.85/0.42,9.45/0.37,10.05/0.33}
            \draw[line width=1.2pt] (\x,0) -- (\x,\h);

        \draw[dotted] (3.35,-0.05) -- (3.35,3.25);
        \node[below] at (1.70,-0.18) {$s_{\ell,n}>\tau_n$};
        \node[below] at (7.10,-0.18) {$s_{\ell,n}\le\tau_n$};

        % Orthogonal decomposition; the widths suggest low versus high dimension.
        \node at (5.35,-1.15) {$\mathbb R^n=\calE_n\oplus  \calD_n$};
        \fill[gray!30] (0.25,-2.35) rectangle (3.35,-1.55);
        \draw (0.25,-2.35) rectangle (3.35,-1.55);
        \draw (3.35,-2.35) rectangle (10.45,-1.55);
        \node at (1.80,-1.95) {$\calE_n$};
        \node at (6.90,-1.95) {$\calD_n$};

        % Minimal treatment labels.
        \draw[->] (1.80,-2.45) -- (1.80,-3.05);
        \draw[->] (6.90,-2.45) -- (6.90,-3.05);
        \node[align=center] at (1.80,-3.52)
            {separate\\Gaussian law};
        \node[align=center] at (6.90,-3.52)
            {Brascamp--Lieb\\projection};
    \end{tikzpicture}
    \caption{Schematic diffuse--exceptional decomposition of the 
    second moment matrix
    $S_{2,\mathrm{rc}}^{(n)}=\sum_{\ell=1}^n
    s_{\ell,n}u_{\ell,n}u_{\ell,n}^{\top}$, with
    $s_{1,n}\ge\ldots\ge s_{n,n}\ge0$.  The exceptional subspace
    $\calE_n$ is spanned by the eigenvectors whose eigenvalues exceed
    $\tau_n$, whereas $ \calD_n= \calE_n^{\perp}$ contains the
    remaining directions.  In the converse proof, the exceptional component
    is controlled by a broadened Gaussian output law (Subsection~\ref{sec:broad}), while the spectral bound
    $\Pi_{\calD_n}S_{2,\mathrm{rc}}^{(n)}\Pi_{\calD_n}
    \preceq\tau_n\Pi_{\calD_n}$ enables the Brascamp--Lieb argument on
    $ \calD_n$ (Subsection~\ref{sec:bl}).}
    \label{fig:diffuse-exceptional}
\end{figure}

\begin{proposition}[Diffuse-exceptional subspace decomposition] \label{prop:diffuse_exc}
    Let 
    \begin{align}
    S_{2,\rc}^{(n)} := \bbE\big[X_{2,\rc}^n(X_{2,\rc}^n)^\top\big]=\sum_\ell s_{\ell,n}u_{\ell,n}u_{\ell,n}^\top    
    \end{align}
    be the second moment matrix of the trimmed user-2 codebook  where $\{u_{\ell,n}\}_{\ell=1}^n$ is an orthonormal eigenbasis. Let $\tau_n:=n^{3/8}$ and define 
    \begin{align}
        \Pi_{\calE_n}:=\sum_{\ell : s_{\ell,n}> \tau_n}u_{\ell,n}u_{\ell,n}^\top,\qquad \calE_n:=\mathrm{ran}(\Pi_{\calE_n}),\qquad \Pi_{\calD_n}=I_n-\Pi_{\calE_n},\qquad \calD_n=\calE_n^\perp.
    \end{align}
    Writing $r_n :=\mathrm{dim}(\calE_n)$, we have
    \begin{align}
        \| \Pi_{\calD_n}S_{2,\rc}^{(n)} \Pi_{\calD_n}\|_{\mathrm{op}} &\le\tau_n ,\label{eqn:bound_by_tau}\\
        r_n &=O(n^{3/8}),\label{eqn:bound_rn}\\
        \bbE \big[ \|\Pi_{\calE_n}X_{2,\rc}^n\|^2\big] &= O(n^{3/4}) ,\label{eqn:X2onE}\\
\bbE \big[ \|\Pi_{\calE_n}X_{1,\rc}^n\|^2\big] &= O(\sqrt{n}) . \label{eqn:X1onE}
    \end{align}
\end{proposition}
\begin{proof}
We first start the proof by deriving some properties of a certain function $g_{q,P}(\lambda )$ which is the gap between the tangent line to the concave function $f_q(\lambda)=\log(1+\lambda/q)$ at $P$ and $f_q(\lambda )$. For fixed $q,P>0$, define the nonnegative function
\begin{align}
    g_{q,P}(\lambda):=\log\left(1+\frac{P}{q} \right)+\frac{\lambda-P}{q+P}-\log\left( 1+\frac{\lambda}{q}\right).
\end{align}
%Concavity of $\lambda\mapsto\log(1+\lambda/q)$ implies that $g_{q,P}(\lambda)\ge0$. 
If $X^n$ satisfies $\|X^n\|^2=nP$ and has mean vector $\mu^{(n)} = \bbE[X^n] $ and covariance matrix $\mathrm{Cov}(X^n)=K^{(n)}$ with eigenvalues $\{\lambda_\ell\}$, define 
\begin{align}
    \Gamma=n\rvC(P/q)-\frac{1}{2}\log\det\bigg( I_n+\frac{ K^{(n)}}{q}\bigg).
\end{align}
Because $\tr( K^{(n)})-nP=-\| \mu^{(n)}\|^2$, summing $g_{q,P}(\lambda_\ell)$ over $\ell$ gives
\begin{align}
    2\Gamma = \bigg(\sum_\ell g_{q,P}(\lambda_\ell) \bigg)+ \frac{ \|\mu^{(n)} \|^2}{q+P}. \label{eqn:sumg}
\end{align}
Consequently, we have the bounds
\begin{align}
    \sum_\ell g_{q,P}(\lambda_\ell)\le 2\Gamma\quad\mbox{and}\quad  \|\mu^{(n)} \|^2\le 2(q+P)\Gamma.  \label{eqn:bounds_g}
\end{align}
Also, $\lim_{\lambda\to\infty}\frac{g_{q,P}(\lambda)}{\lambda}=\frac{1}{q+P}$.
Hence, there is a finite $\lambda_0(q,P)\ge0$ such that for all
\begin{align}
     \lambda\ge\lambda_0(q,P)\quad \Longrightarrow \quad g_{q,P}(\lambda)\ge\frac{\lambda}{2(q+P)}. \label{eqn:bounds_g_large}
\end{align}
There is also a finite $c_{q,P}$ such that  
\begin{align}
    \lambda\le c_{q,P}\big(1+ g_{q,P}(\lambda) \big)\qquad\forall\, \lambda\ge0. \label{eqn:bounds_g_small}
\end{align}
By considering $0\le\lambda\le\lambda_0(q,P)$ and $\lambda>\lambda_0(q,P)$, we see that one may take $c_{q,P}=\max\{\lambda_0(q,P),2(q+P)\}$.

Now, we start the proof proper. The spectral definition of $S_{2,\rc}^{(n)}$ in the statement of the proposition directly gives $\| \Pi_{\calD_n}S_{2,\rc}^{(n)} \Pi_{\calD_n}\|_{\mathrm{op}}\le \tau_n$, which is~\eqref{eqn:bound_by_tau}.

For any vector $x^n\in\bbR^n$, define their projections onto the subspaces $\calE_n$ and $\calD_n$ respectively as 
\begin{equation}
x_{\calE}^n = \Pi_{\calE_n}x^n\quad \mbox{and}\quad  x_{\calD}^n = \Pi_{\calD_n}x^n. \label{eqn:convention} 
\end{equation}
Now, we define the expectations of the exceptional energies
\begin{align}
    \bara_1^{(n)} &:=\bbE \big[ \|X_{1,\rc,\calE}^n\|^2 \big], \\
    \bara_2^{(n)} &:=\bbE \big[ \|X_{2,\rc,\calE}^n\|^2 \big]=\tr(\Pi_{\calE_n}S_{2,\rc}^{(n)}).
\end{align}
Define the high-eigenvalue part of the covariance matrix $K_{2,\rc}^{(n)} = S_{2,\rc}^{(n)}-\mu_{2,\rc}^{(n)}(\mu_{2,\rc}^{(n)})^\top = \sum_\ell \lambda_{\ell,n}w_{\ell,n}w_{\ell,n}^\top$ as 
\begin{align}
K_{2,>,\rc}^{(n)}:=\sum_{\ell:\lambda_{\ell,n}>\tau_n /2}\lambda_{\ell,n}w_{\ell,n}w_{\ell,n}^\top .\label{eqn:Khigh}
\end{align}
The corresponding low-eigenvalue part of $K_{2,\rc}^{(n)} $ is denoted as $K_{2,\le,\rc}^{(n)} := K_{2,\rc}^{(n)}-K_{2,>,\rc}^{(n)}$. By the definition  of $K_{2,>,\rc}^{(n)}$ in \eqref{eqn:Khigh}, we see that $0\preceq K_{2,\le,\rc}^{(n)} \preceq (\tau_n/2) I_n$. Therefore,
\begin{align}
    \bara_2^{(n)} & = \tr(\Pi_{\calE_n}K_{2,\rc}^{(n)}  ) + \|\Pi_{\calE_n}\mu_{2,\rc}^{(n)} \|^2 \\ 
    &\le\frac{r_n\tau_n}{2 } + \tr( K_{2,>,\rc}^{(n)}) + \|\mu_{2,\rc}^{(n)} \|^2 , \label{eqn:a2_part1}
\end{align}
where the first term uses $\tr(\Pi_{\calE_n}K_{2,\le,\rc}^{(n)})\le r_n\tau_n/2$   and the high-spectrum and projected-mean terms
were enlarged using positive semi-definiteness and contraction of orthogonal projection, respectively. If $r_n>0$, 
\begin{align}
    \bara_2^{(n)} =\sum_{\ell : s_{\ell,n}>\tau_n} s_{\ell,n}> \tau_n \big| \{\ell: s_{\ell,n}>\tau_n\} \big| = r_n\tau_n.\label{eqn:a2_part2}
\end{align}
Hence if $r_n>0$, combining \eqref{eqn:a2_part1} and \eqref{eqn:a2_part2} yields
\begin{align}
\bara_2^{(n)}\le 2\big(   \tr( K_{2,>,\rc}^{(n)}) + \|\mu_{2,\rc}^{(n)} \|^2 \big). \label{eqn:a2_bd}
\end{align}
When $r_n=0$, \eqref{eqn:a2_bd} is immediate because $\bara_2^{(n)}=0$. For large $n$, every nonzero eigenvalue of $K_{2,>,\rc}^{(n)}$ exceeds the constant $\lambda_0(q_1, P_2)$. The bounds in \eqref{eqn:bounds_g} and \eqref{eqn:bounds_g_large} imply that 
\begin{align}
    \tr( K_{2,>,\rc}^{(n)}) \le 2(q_1+P_2)\sum_{\ell:\lambda_{\ell,n}>\tau_n/2} g_{q_1,P_2}(\lambda_\ell)\le 4(q_1+P_2)\Gamma_{2,n}\quad \mbox{and}\quad \|\mu_{2,\rc}^{(n)}\|^2 \le 2(q_1+P_2)\Gamma_{2,n}.
\end{align}
Therefore,
\begin{align}
    \bara_2^{(n)}\le 12(q_1+P_2)\Gamma_{2,n} = O(n^{3/4}),
\end{align}
which proves \eqref{eqn:X2onE}. Moreover, 
\begin{align}
    r_n \le \frac{\bara_2^{(n)}}{\tau_n}=  O(n^{3/4})  n^{- 3/8} =  O(n^{3/8}), \label{eqn:bound_rn_prf}
\end{align}
which proves \eqref{eqn:bound_rn}. 

Finally, for user-1, writing the eigenvalues of $K_{1,\rc}^{(n)}$ in nonincreasing order  as $\lambda_1^{\downarrow}(K_{1,\rc}^{(n)})\ge \lambda_2^{\downarrow}(K_{1,\rc}^{(n)})\ge\ldots\ge \lambda_n^{\downarrow}(K_{1,\rc}^{(n)})\ge0$,  Ky Fan's maximum principle (see, e.g., Bhatia~\cite{Bhatia1997}) and  \eqref{eqn:bounds_g_small} give
\begin{align}
    \bara_1^{(n)} &= \tr(\Pi_{\calE_n}K_{1,\rc}^{(n)})+ \| \mu_{1,\rc}^{(n)}\|^2\nn\\
    &\le \sum_{\ell=1}^{r_n}\lambda_\ell^{\downarrow}(K_{1,\rc}^{(n)}) + \| \mu_{1,\rc}^{(n)}\|^2 \nn\\
    &\le c_{1,P_1}r_n+c_{1,P_1}\sum_{\ell=1}^ng_{1,P_1} (\lambda_\ell (K_{1,\rc}^{(n)}) ) + \| \mu_{1,\rc}^{(n)}\|^2 = O(\sqrt{n}).
\end{align}
Here, the final equality first uses  \eqref{eqn:sumg} to claim that the final two terms is $O(\Gamma_{1,n})$, the fact that log-determinant deficit result for user-1  (i.e., $\Gamma_{1,n}=O(\sqrt{n})$), and finally~\eqref{eqn:bound_rn_prf}. This  proves \eqref{eqn:X1onE} and hence, Proposition~\ref{prop:diffuse_exc}.
\end{proof}

\subsection{Exceptional energy expurgation} \label{sec:energy_expurg}
Proposition~\ref{prop:diffuse_exc} provides control of the average energies of the codewords on the exceptional subspace. The subsequent arguments require pointwise control: the Brascamp--Lieb step normalizes every diffuse user-2
codeword and therefore needs its diffuse norm to remain of order $\sqrt{n}$, while the broad Gaussian output distributions on $\calE_n$ require a common deterministic energy bound for every retained codeword pair. We therefore remove the vanishing fractions of messages whose $\calE_n$-energies are atypically large. The two message sets are trimmed separately so that the retained code remains a Cartesian product and the two random codewords remain independent.

In the following, the subspaces $\calE_n$ and $\calD_n$ are those constructed from the row/column-trimmed user-2 codebook in Proposition~\ref{prop:diffuse_exc}; they are not recomputed after the present expurgation.

\begin{proposition}[Final product code] \label{prop:final}
    Set $b_n:=n^{7/8}$ and retain, separately for each $ j \in \{1,2\}$,\footnote{The superscript $^\rmf$ stands for ``final''.}
    \begin{align}
        \calM_{j,n}^{\rmf} := \big\{ m_j\in \calM_{j,n}^{\rc}: \| \Pi_{\calE_n}x_j^n(m_j) \|^2\le b_n \big\}.
    \end{align}
    Define the additional discarded fractions 
    \begin{align}
        \delta_{j,n}:= 1- \frac{|\calM_{j,n}^{\rmf} |}{ |\calM_{j,n}^{\rc} |},\qquad j \in \{1,2\},
    \end{align}
    and let $\eps_n'$ be the average error restricted to the resulting product code. Let $W_j^{\rmf}$ be independent and uniform on $\calM_{j,n}^{\rmf} $ and set $X_{j,\rmf}^n:= x_j^n(W_j^{\rmf})$ and define $S_{2,\rmf}^{(n)} := \bbE[ X_{2,\rmf}^n(X_{2,\rmf}^n)^\top]$. Then the message sizes and average error probability after expurgation satisfy
    \begin{align}
\delta_{1,n}
    &=O(n^{-3/8}), \label{eqn:retain1}\\
\delta_{2,n}
    &=O(n^{-1/8}),                    \label{eqn:retain2}                  \\
\log|\calM_{j,n}^{\rmf}|
    &=\log M_{j,n}+O(1), \label{eqn:sizef}
\qquad
    j\in\{1,2\},                                       \\
\limsup_{n\to\infty}\eps_n' \label{eqn:error_f}
    &\le\eps.                                           
\end{align}
The pointwise energies of the final codebooks satisfy
\begin{align}
\|\Pi_{\calE_n}x_j^n(m_j)\|^2
    &\le b_n,
    \qquad j\in\{1,2\},\quad
      m_j\in\calM_{j,n}^{\rmf},            \label{eqn:energyE}             \\
\max_{m_j\in\calM_{j,n}^{\rmf}}
\left|
    \frac{\|\Pi_{\calD_n}x_j^n(m_j)\|^2}{n}-P_j
\right|
    &\to 0,
    \qquad  j\in\{1,2\}.  \label{eqn:energyD}
\end{align}
The  final code additionally satisfies
\begin{align}
\big\|
    \Pi_{\calD_n}S_{2,\rmf}^{(n)}\Pi_{\calD_n}
\big\|_{\mathrm{op}}
    &\le\frac{\tau_n}{1-\delta_{2,n}},          \label{eqn:distri1}        \\
D\big(P_{X_{1,\rmf}^n+Z^n} \,\|\, \calN(0,q_1I_n) \big) &=O(\sqrt n).\label{eqn:distri2}
\end{align}
\end{proposition}
\begin{proof}
    For each $m_1\in\calM_{1,n}^\rc$ and $m_2\in\calM_{2,n}^\rc$, define the energies  projected onto the subspace $\calE_n$ as
\begin{align}
    a_j^{(n)}(m_j) := \|\Pi_{\calE_n}x_j^n(m_j) \|^2,\qquad j \in \{1,2\}.
\end{align}
Retain messages as stated in the proposition, i.e., 
\begin{align}
    \calM_{j,n}^\rmf := \big\{ m_j \in \calM_{j,n}^\rc: a_j^{(n)} (m_j)\le b_n\big\}. \label{eqn:energy_exp}
\end{align}
Markov's inequality then gives
\begin{align}
    \delta_{1,n}\le\frac{\bara_1^{(n)}}{b_n} = O(n^{-3/8})\quad \mbox{and}\quad\delta_{2,n}\le\frac{\bara_2^{(n)}}{b_n} = O(n^{-1/8}),
\end{align}
which proves~\eqref{eqn:retain1} and~\eqref{eqn:retain2}. Since a vanishing fraction of messages is removed (recall that $\delta_{j,n}\to0$ for $j \in \{1,2\}$),~\eqref{eqn:sizef} follows similarly as Part (i) of Proposition~\ref{prop:trim1}. In particular, neither second-order coefficient is affected. Map every output of the original decoder outside
$\calM_{1,n}^{\rmf}\times\calM_{2,n}^{\rmf}$ to an arbitrary fixed
retained pair, and let $\eps_n'$ denote the average error probability of
the resulting decoder. Then
\begin{align}
\eps_n'
&\le
\frac{1}{|\calM_{1,n}^{\rmf}||\calM_{2,n}^{\rmf}|}
\sum_{m_1\in\calM_{1,n}^{\rmf}}
\sum_{m_2\in\calM_{2,n}^{\rmf}}
p_n(m_1,m_2) \nonumber\\
&\le
\frac{\bar\eps_n^{\rc}}
     {(1-\delta_{1,n})(1-\delta_{2,n})}
\le
\frac{\eps_n}
     {(1-\delta_{1,n})(1-\delta_{2,n})}
=\eps_n+o(1),
\end{align}
which proves~\eqref{eqn:error_f}.

The bound in~\eqref{eqn:energyE} follows directly from the construction of $\calM_{j,n}^\rmf$ for each $j \in \{1,2\}$. In addition, for each user $j$ and every $m_j\in\calM_{j,n}^\rmf$, the exact-shell property (Proposition~\ref{prop:spherical}) gives 
\begin{align}
    \|\Pi_{\calD_n}x_j^n(m_j)\|^2=nP_j- a_j^{(n)}(m_j),
\end{align}
and hence, uniformly over messages 
\begin{align}
    P_j - \frac{b_n}{n}\le\frac{\|\Pi_{\calD_n}x_j^n(m_j)\|^2 }{n}\le P_j,
\end{align}
which proves~\eqref{eqn:energyD}.

Since a fraction $1-\delta_{2,n}$ of messages was retained in user-2's codebook,
\begin{align}
    S_{2,\rmf}^{(n)} &= \frac{1}{|\calM_{2,n}^{\rmf}|}\sum_{m_2 \in \calM_{2,n}^{\rmf}} x_2^n(m_2)  x_2^n(m_2)^\top\nn \\
    &= \frac{1}{(1-\delta_{2,n})|\calM_{2,n}^{\rc}|}\sum_{m_2 \in \calM_{2,n}^{\rmf}} x_2^n(m_2) x_2^n(m_2)^\top \nn\\
    &\preceq \frac{1}{(1-\delta_{2,n})|\calM_{2,n}^{\rc}|}\sum_{m_2 \in \calM_{2,n}^{\rc}} x_2^n(m_2)x_2^n(m_2)^\top = \frac{S_{2,\rc}^{(n)}}{1-\delta_{2,n}}.
\end{align}
The inequality holds in the positive-semidefinite  order because every omitted outer
product is positive semidefinite. Left and right multiplication by $\Pi_{\calD_n}$ preserves
this order; hence
\begin{align}
\Pi_{\calD_n}S_{2,\rmf}^{(n)}\Pi_{\calD_n}
&\preceq
\frac{
\Pi_{\calD_n}S_{2,\rc}^{(n)}\Pi_{\calD_n}}
{1-\delta_{2,n}} 
\preceq
\frac{\tau_n}{1-\delta_{2,n}}\Pi_{\calD_n}.
\end{align}
Taking operator norms and using~\eqref{eqn:bound_by_tau} proves~\eqref{eqn:distri1}. For
every $m_1\in \calM_{1,n}^{\rmf}\subseteq \calM_{1,n}^{\rc}$, the stochastic decoder in \eqref{eqn:varphi_m1} still has success probability at least $1-t$.
Restricting its decision functions to $\calM_{1,n}^{\rmf}$  preserves subnormalization, and any leftover probability may be assigned to one retained message. The decoder continues to randomize over the original full user-2 codebook, so the user-2 expurgation does not affect this maximal-error guarantee. Lemma~\ref{lem:good_code}
and \eqref{eqn:sizef} therefore give~\eqref{eqn:distri2}.\end{proof}

\subsection{The Brascamp--Lieb projection step} \label{sec:bl}
Henceforth, we abbreviate $X_1^n := X_{1,\rmf}^n$ and  $X_2^n := X_{2,\rmf}^n$. All
subsequent probabilities and expectations use these independent final message-uniform laws and  independent Gaussian noise $Z^n\sim\calN(0,I_n)$.

The diffuse operator bound in \eqref{eqn:distri1} from Proposition \ref{prop:final} ensures that the random user-2 direction is sufficiently spread out for a one-dimensional projection of the user-1 output to inherit its Gaussianity. This section's code-specific conclusion is the following central limit theorem (CLT).  
\begin{proposition}[Diffuse-overlap CLT]\label{prop:diffuse_overlap}
For the independent final message-uniform codewords,
\begin{align}
    T_{\calD,n}:= \frac{\langle  X_{1,\calD}^n,  X_{2,\calD}^n\rangle}{\sqrt{n}} \Rightarrow \calN(0,P_1P_2), \label{eqn:conv_TD}
\end{align}
where we use the notational convention to denote the projection onto the subspace $\calD_n$ in \eqref{eqn:convention}.
\end{proposition}
The following lemma is an operator norm corollary of the
entropy/relative-entropy formulation of the geometric
Brascamp--Lieb inequality \cite{CarlenCorderoErausquin2009,LiuCourtadeCuffVerdu2016}.
We include the short proof because the formulation in terms of
$\|\mathbb E[VV^\top]\|_{\mathrm{op}}$ is particularly convenient
for the present application.
\begin{lemma}[Gaussian-relative entropy projection inequality] \label{lem:RE_proj}
    Let $U$ be a random vector in a $d$-dimensional Euclidean space. Let $V$, independent of $U$, be a finitely supported random unit vector. Put 
    \begin{align}
        F:=\mathbb E[VV^\top],\qquad\lambda= \| F \|_{\mathrm{op}}.
    \end{align}
    Then for every $q>0$,
    \begin{align}
        \bbE_V \Big[ D\big(P_{ \langle U,V \rangle | V} \|\calN(0,q) \big) \Big]\le \lambda \,   D\big(P_U\| \calN(0,q I_d) \big). \label{eqn:bl2}
    \end{align} 
\end{lemma}
\begin{proof}[Proof of Lemma~\ref{lem:RE_proj}]
    Let the support of $V$ be $\{v_k\}$ and put
$p_k:=\Pr(V=v_k)$. Since $V$ is independent of $U$, the conditional
law of $\langle V,U\rangle$ given $V=v_k$ is
$P_{\langle v_k,U\rangle}$. Thus, the left-hand side of
\eqref{eqn:bl2} is
\begin{align}
 \bbE_V \Big[ D\big(P_{ \langle U,V \rangle | V} \|\calN(0,q) \big) \Big]=\sum_k p_k
D\left(P_{\langle v_k,U\rangle}\,\middle\|\,\calN(0,q)\right).
\label{eqn:conditional-projection-sum}
\end{align} Since $\tr(F)=1$, one has $\lambda>0$ and $F/\lambda\preceq I_d$. Spectrally decompose
    \begin{align}
        I_d-\frac{F}{\lambda}=\sum_\ell\omega_\ell u_\ell u_\ell^\top ,\qquad\omega_\ell\ge0.
    \end{align}
    Thus,
    \begin{align}
        \sum_k \frac{p_k}{\lambda} v_k v_k^\top + \sum_\ell\omega_\ell u_\ell u_\ell^\top =I_d. \label{eqn:spectral2}
    \end{align}
    For any finite tight frame $\{(e_s, w_s)\}$ of unit directions $e_s$ and weights $w_s$ satisfying $\sum_s w_s e_s e_s^\top=I_d$, the functional form of the geometric Brascamp–Lieb inequality, with $G\sim \calN(0,qI_d) $ and $g\sim \calN(0,q)$, is \cite{CarlenCorderoErausquin2009}
    \begin{align}
        \log\bbE\bigg[\exp\bigg( \sum_s w_sf_s\big( \langle e_s, G\rangle\big) \bigg)\bigg]\le\sum_s w_s\log\bbE\big[e^{f_s(g)} \big] \label{eqn:bl_functional}
    \end{align}
    for all bounded measurable $f_s$. 

If $D(P_U\|\calN(0,qI_d))=\infty$, the desired inequality is immediate. Suppose therefore that it is finite. Apply the Donsker--Varadhan
variational formula~\cite{DonskerVaradhan1975,DupuisEllis1997} with
$
h(u):=\sum_sw_sf_s(\langle e_s,u\rangle).
$ 
Together with \eqref{eqn:bl_functional}, this gives
\begin{align}
D\left(P_U\,\middle\|\,\calN(0,qI_d)\right) &\ge
\mathbb E\bigg[
\sum_sw_sf_s(\langle e_s,U\rangle)
\bigg]
-
\log\bbE\bigg[\exp\bigg( \sum_s w_sf_s\big( \langle e_s, G\rangle\big) \bigg)\bigg]
\nonumber\\
&\ge
\sum_sw_s
\left\{
\mathbb E\left[f_s(\langle e_s,U\rangle)\right]
-
\log\mathbb E\big[e^{f_s(g)}\big]
\right\}.
\label{eqn:dv-bl}
\end{align}
Taking the supremum independently over the finitely many functions
$f_s$ yields
\begin{align}
D\left(P_U\,\middle\|\,\calN(0,qI_d)\right)
\ge
\sum_sw_s
D\left(
P_{\langle e_s,U\rangle}
\,\middle\|\,
\calN(0,q)
\right).
\label{eqn:take_sup}
\end{align}
Apply \eqref{eqn:take_sup} to the identity decomposition in~\eqref{eqn:spectral2}. We obtain
\begin{align}
D\left(P_U\,\middle\|\,\calN(0,qI_d)\right)
&\ge
\sum_k\frac{p_k}{\lambda}
D\left(
P_{\langle v_k,U\rangle}
\,\middle\|\,
\calN(0,q)
\right)
 +
\sum_\ell\omega_\ell
D\left(
P_{\langle u_\ell,U\rangle}
\,\middle\|\,
\calN(0,q)
\right).
\end{align}
The second sum is nonnegative. Discarding it and multiplying by
$\lambda$ throughout gives
\begin{align}
\sum_kp_k
D\left(
P_{\langle v_k,U\rangle}
\,\middle\|\,
\calN(0,q)
\right)
\le
\lambda\,
D\left(P_U\,\middle\|\,\calN(0,qI_d)\right).
\end{align}
Together with \eqref{eqn:conditional-projection-sum}, this proves
\eqref{eqn:bl2}.
\end{proof}

\begin{proof}[Proof of Proposition~\ref{prop:diffuse_overlap}]
    Define the final random unit vector
    \begin{align}
        V_n := \frac{X_{2,\calD}^n}{\|X_{2,\calD}^n\|},
    \end{align}
    which is well-defined for large enough $n$ by \eqref{eqn:energyD}. Let $F_n = \bbE [ V_n V_n^\top]$. The bounds in \eqref{eqn:energyD} and  \eqref{eqn:distri1} give
    \begin{align}
        V_n V_n^\top &= \frac{X_{2,\calD}^n(X_{2,\calD}^n)^\top}{\|X_{2,\calD}^n\|^2}\preceq \frac{X_{2,\calD}^n(X_{2,\calD}^n)^\top}{nP_2-b_n} ,\\
        \bbE[ X_{2,\calD}^n(X_{2,\calD}^n)^\top]  &= \Pi_{\calD_n} S_{2,\rmf}^{(n)}\Pi_{\calD_n},\\
        F_n& \preceq\frac{ \Pi_{\calD_n} S_{2,\rmf}^{(n)}\Pi_{\calD_n}}{nP_2-b_n},\\
        \lambda_n &:=\|F\|_{\mathrm{op}}\le\frac{\tau_n}{(1-\delta_{2,n})(nP_2-b_n)}= O(n^{-5/8}).
    \end{align}
    Set $U_\calD^n=\Pi_{\calD_n}(X_1^n+Z^n)$.  Applying the data processing inequality to \eqref{eqn:distri2} gives 
    \begin{align}
        D \big(P_{U_\calD^n}\|\calN(0,q_1 I_{\calD_n})\big)=O(\sqrt{n}).
    \end{align}
    The random variables $V_n$ and $U_\calD^n$ are independent. Since $r_n=O(n^{3/8})$ one also has $d_n=n-r_n\ge 1$ for all sufficiently large $n$, so Lemma~\ref{lem:RE_proj} applies on $\calD_n$ and therefore yields
    \begin{align}
        \calD_n^{\mathrm{proj}}=\bbE_{V_n}\big[ D(P_{\langle V_n ,U_\calD^n\rangle  |V_n} \| \calN(0,q_1) \big]\le O(n^{-5/8}) O(n^{1/2})=O(n^{-1/8})\to0.
    \end{align}
    Define 
    \begin{align}
        \Xi_n:=\langle V_n, X_{1,\calD}^n \rangle \quad\mbox{and}\quad G_n=\langle V_n, Z_\calD^n \rangle.
    \end{align}
Because $Z^n$ is isotropic and independent of both codebooks, $G_n \sim \calN(0,1)$ and is independent of $(V_n, \Xi_n)$. Moreover $\langle V_n ,U_\calD^n\rangle=\Xi_n+G_n$. Writing $\tilY_n=\Xi_n+G_n$, the relative entropy chain rule gives
\begin{align}
    \calD_n^{\mathrm{proj}}=D\big( P_{V_n,\tilY_n}\|P_{V_n}\times \calN(0,q_1) \big)=I(V_n;\tilY_n)+D\big(P_{\tilY_n}\|\calN(0,q_1) \big).
\end{align}
Consequently, 
\begin{align}
    D\big(P_{\Xi_n+G_n}\| \calN(0,q_1) \big)\le \calD_n^{\mathrm{proj}}\to0.
\end{align}
Pinsker's inequality gives $d_{\mathrm{TV}}\left(P_{\Xi_n+G_n},\calN(0,q_1) \right)\to0,$ and hence, $ \Xi_n+G_n\Rightarrow\calN(0,q_1)$.
Recall that $G_n\sim\calN(0,1)$ and is independent of $\Xi_n$.
Therefore, for every $t\in\bbR$, a moment generating function computation yields $\varphi_{\Xi_n}(t)e^{-t^2/2}
=\varphi_{\Xi_n+G_n}(t)
\to e^{-q_1t^2/2}.$ 
Since $q_1=1+P_1$, division by $e^{-t^2/2}>0$ gives
$\varphi_{\Xi_n}(t)\to e^{-P_1t^2/2}.$  L\'evy's continuity theorem consequently yields
$\Xi_n\Rightarrow\calN(0,P_1)$.
By definition, 
\begin{align}
    T_{\calD,n}= \frac{\|X_{2,\calD}^n \|}{\sqrt{n}} \Xi_n  \quad\mbox{and}\quad\frac{\|X_{2,\calD}^n \|}{\sqrt{n}} \xrightarrow{\mathrm P} \sqrt{P_2}
\end{align}
uniformly over user-2's messages. Hence, Slutsky's theorem yields~\eqref{eqn:conv_TD}. 
\end{proof}

\subsection{Broad Gaussian control of the exceptional subspace} \label{sec:broad}
The projections of the two codewords onto the exceptional subspace $\calE_n$ need not have an asymptotically Gaussian inner product, and their energies need not vanish.
However, $\calE_n$ has low dimension and, after expurgation, every projected
codeword satisfies a common deterministic energy bound per \eqref{eqn:energyE}. We therefore use on $\calE_n$ a Gaussian output distribution with an enlarged variance chosen
to accommodate this bound. The resulting contribution to the information
density can be negative, but negative values only help the converse by making
the Verd\'u--Han threshold event less likely. It is therefore sufficient to
control only its upper tail though a Chebyshev argument.

\begin{lemma}[Broad Gaussian Likelihood] \label{lem:broad_gauss}
    Let $\calE_n$ have dimension $r\ge1$, let $B\ge 0$ and suppose $x^r\in \calE_n$ satisfies $\|x^r\|^2\le B$. Let $Z^r\sim\calN(0,I_{\calE_n})$, $Y^r=x^r+Z^r$ and 
    \begin{align}
        q_B (y^r)= \calN\left(y^r;0,\Big( 1+\frac{B}{r}\Big) I_{\calE_n}\right).
    \end{align}
    Then, the log-likelihood ratio 
    \begin{align}
        \Lambda_B(x^r,Z^r) &= \log\frac{\calN(Y^r;x^r, I_{\calE_n})}{q_B  (Y^r)} \nn\\
        &=r\rvC(B/r)+ \frac{ \|x^r\|^2-B}{2(1+B/r)}+ \frac{\langle x^r,Z^r\rangle}{1+B/r}- \frac{B/r}{2(1+B/r)} \big(\|Z^r\|^2-r\big) \label{eqn:Lambda_B}
    \end{align}
    satisfies
    \begin{align}
        \bbE[ \Lambda_B(x^r,Z^r)]\le r\rvC(B/r)\quad\mbox{and}\quad \mathrm{Var}( \Lambda_B(x^r,Z^r))\le\frac{r}{2}.
    \end{align}
\end{lemma}

\begin{proof}[Proof of Lemma~\ref{lem:broad_gauss}]
    The expansion in  \eqref{eqn:Lambda_B} follows by direct calculation involving Gaussian densities. The expectation is at most $r\rvC(B/r)$ because the second term is non-positive and the  third and fourth  terms  have zero mean.  The linear Gaussian term and centered quadratic term are uncorrelated.  Hence, writing $s=B/r$, we have
    \begin{align}
        \mathrm{Var}( \Lambda_B(x^r,Z^r))&=\frac{\|x^r\|^2}{(1+s)^2}+\frac{s^2}{4(1+s)^2}\mathrm{Var} ( \|Z^r\|^2)\\
        &\le r\frac{s+s^2/2}{(1+s)^2}\le\frac{r}{2},
    \end{align}
    where the last inequality is equivalent to $2s+s^2\le (1+s)^2$.
\end{proof}

\begin{proposition} \label{prop:excep_likeli}
    For $\hatb_n \in \{b_n,4b_n\}$, let $Z_\calE^n\sim\calN(0,I_{\calE_n})$. With $\Lambda_{\hatb_n}$ defined in Lemma~\ref{lem:broad_gauss}  and interpreted as zero when $r_n=0$, for every $\eta>0$,
    \begin{align}
        \sup_{x^n \in \calE_n : \|x^n\|^2\le\hatb_n}\Pr\big( \Lambda_{\hatb_n}(x^n,Z_\calE^n)^+>\eta\sqrt{n} \big)\to0. \label{eqn:prop_cheby}
    \end{align}
    Here and in the following, $u^+:=\max\{u,0\}$. 
\end{proposition}

\begin{proof}[Proof of Proposition~\ref{prop:excep_likeli}]
Recall that $b_n=n^{7/8}$. If $r_n=0$, then $\calE_n=\{0\}$ and the
exceptional likelihood is defined to be zero. It therefore suffices to
consider indices for which $r_n\ge1$.

For $B>0$, define $f_B(r):=r\rvC(B/r)$ for $ r>0$.
This function is nondecreasing because
\begin{align}
f_B'(r)
=\frac12\left[
\log\left(1+\frac Br\right)-\frac{B}{B+r}
\right]\ge0.
\end{align}
Choose a finite constant $K_0$ such that
$r_n\le K_0n^{3/8}$ for all sufficiently large $n$, and put $R_n:=K_0n^{3/8}$.
For either choice $B\in\{b_n,4b_n\}$, monotonicity gives
\begin{align}
\Upsilon_n(B)
&:=r_n\rvC(B/r_n)\le R_n\rvC(B/R_n)
\le
\frac{K_0n^{3/8}}{2}
\log\left(1+\frac{4b_n}{K_0n^{3/8}}\right)=O(n^{3/8}\log n)
=o(\sqrt n).
\label{eqn:broad-cost}
\end{align}

For an   $x^n\in\calE_n$ with $\|x^n\|^2\le B$, define its expected log-likelihood by
\begin{align}
\mu_{n,B}(x^n)
:=\mathbb E[\Lambda_B(x^n,Z_{\calE}^n)]=
r_n\rvC(B/r_n)
+\frac{\|x^n\|^2-B}{2(1+B/r_n)}.
\end{align}
Lemma~\ref{lem:broad_gauss} gives, uniformly over all $x^n\in\calE_n $ such that $\|x^n\|^2\le B$,
\begin{align}
\mu_{n,B}(x^n)\le \Upsilon_n(B),
\quad\mbox{and}\quad
\mathrm{Var}(\Lambda_B(x^n,Z_{\calE}^n))\le\frac{r_n}{2}.
\end{align}

Fix $\eta>0$. Since $\Upsilon_n(B)=o(\sqrt n)$, we have
$\eta\sqrt n-\Upsilon_n(B)>0$ for all sufficiently large $n$. Moreover,
because $\eta\sqrt n>0$, $ 
\{\Lambda_B(x^n,Z_{\calE}^n)^+>\eta\sqrt n\}
=
\{\Lambda_B(x^n,Z_{\calE}^n)>\eta\sqrt n\}.
$ 
Therefore, Chebyshev's inequality gives
\begin{align}
&\sup_{x^n\in\calE_n: \|x^n\|^2\le B}
\Pr\left(
\Lambda_B(x^n,Z_{\calE}^n)^+>\eta\sqrt n
\right)
\nonumber\\
&\quad\le
\sup_{x^n\in\calE_n: \|x^n\|^2\le B}
\frac{\mathrm{Var}(\Lambda_B(x^n,Z_{\calE}^n))}
     {(\eta\sqrt n-\mu_{n,B}(x^n))^2}
\nonumber\\
&\quad\le
\frac{r_n/2}{(\eta\sqrt n-\Upsilon_n(B))^2}
=O(n^{-5/8}) \to0.
\end{align}
This proves the assertion for both $B=b_n$ and $B=4b_n$.
\end{proof}

% Required in the main document preamble:
% \usepackage{tikz}
% \usetikzlibrary{arrows.meta,positioning}

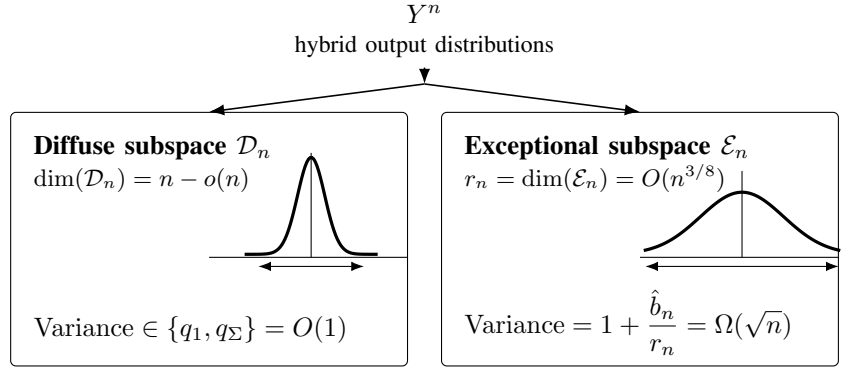
\begin{figure}[t]
\centering
\begin{tikzpicture}[
    >=Latex,
    font=\small,
    branch/.style={draw, rounded corners=2pt, minimum width=5.25cm,
                   minimum height=3.35cm, inner sep=7pt},
    flow/.style={->, semithick},
    axis/.style={thin},
    curve/.style={very thick},
]

% Input and orthogonal split
\node[align=center] (output) at (0,2.75)
    {$Y^n$\\[-1pt]\footnotesize hybrid output distributions};
\coordinate (fork) at (0,2.05);
\draw[flow] (output.south) -- (fork);

% Two unshaded, side-by-side subspace panels
\node[branch] (diffuse) at (-2.85,0) {};
\node[branch] (exceptional) at (2.85,0) {};

\draw[flow] (fork) -- (diffuse.north);
\draw[flow] (fork) -- (exceptional.north);

% Diffuse labels
\node[anchor=north west, align=left]
    at ([xshift=5pt,yshift=-5pt]diffuse.north west)
    {\textbf{Diffuse subspace $\mathcal D_n$}\\
     \footnotesize $\dim(\mathcal D_n)=n-o(n)$};
\node[anchor=south west, align=left]
    at ([xshift=5pt,yshift=5pt]diffuse.south west)
    {$\operatorname{Variance} \in \{q_1,q_\Sigma\}=O(1)$};

% Narrow Gaussian in the diffuse panel
\begin{scope}[shift={([xshift=1.35cm,yshift=0.42cm]diffuse.center)}]
  \draw[axis] (-1.35,-0.66) -- (1.27,-0.66);
  \draw[axis] (0,-0.66) -- (0,0.72);
  \draw[curve, domain=-0.88:0.88, samples=80, variable=\x]
    plot ({\x},{-0.62+1.28*exp(-\x*\x/(2*0.18*0.18))});
   \draw[<->, thin] (-.7,-0.78) -- (.7,-0.78);
\end{scope}

% Exceptional labels
\node[anchor=north west, align=left]
    at ([xshift=5pt,yshift=-5pt]exceptional.north west)
    {\textbf{Exceptional subspace $\mathcal E_n$}\\
     \footnotesize $r_n=\dim(\mathcal E_n)=O(n^{3/8})$};
\node[anchor=south west, align=left]
    at ([xshift=5pt,yshift=.5pt]exceptional.south west)
    {$\displaystyle \operatorname{Variance}=1+\frac{\hat b_n}{r_n}
       =\Omega(\sqrt n)$};

% Wide Gaussian in the exceptional panel
\begin{scope}[shift={([xshift=1.35cm,yshift=0.42cm]exceptional.center)}]
  \draw[axis] (-1.35,-0.66) -- (1.27,-0.66);
  \draw[axis] (0,-0.66) -- (0,0.48);
  \draw[curve, domain=-1.30:1.30, samples=100, variable=\x]
    plot ({\x},{-0.62+0.82*exp(-\x*\x/(2*0.55*0.55))});
  \draw[<->, thin] (-1.28,-0.78) -- (1.28,-0.78);
\end{scope}

\end{tikzpicture}
\caption{Schematic of the hybrid  output distributions~\eqref{eqn:out1} and \eqref{eqn:out_sum} after the orthogonal
decomposition $\mathbb R^n=\mathcal D_n\oplus\mathcal E_n$. On the
high-dimensional diffuse subspace, the   Gaussian has coordinate
variance $q_1$ or $q_\Sigma$, both of constant order. On the low-dimensional
exceptional subspace, its coordinate variance is enlarged to
$1+\hat b_n/r_n$, where $\hat b_n=b_n$ for the user-1 likelihood and
$\hat b_n=4b_n$ for the sum-rate likelihood. Since
$b_n=n^{7/8}$ and $r_n=O(n^{3/8})$, this exceptional variance is at least of
order $\sqrt n$.  }
\label{fig:hybrid-output-laws}
\end{figure}
\subsection{Hybrid  output distributions and exact information densities} 
\label{sec:info_dens}
The broad-Gaussian estimate is one-sided, and so is the conclusion needed for the converse. The hybrid   output distributions  defined in \eqref{eqn:out1} and \eqref{eqn:out_sum} below retain the standard Gaussian laws on $\calD_n$ but widen them on $\calE_n$. See Fig.~\ref{fig:hybrid-output-laws} for a schematic. 

\begin{proposition}[Hybrid information density domination]\label{prop:info_dens}
Let $d_n=n-r_n$ and $Z_\calD^n=\Pi_{\calD_n}Z^n$, and put
\begin{align}
    A_{1,\calD,n} &:= \frac{\langle X_{1,\calD}^n , Z_\calD^n\rangle - \frac{P_1}{2} (\|Z_\calD^n\|^2-d_n)}{q_1\sqrt{n}} ,\\
    A_{\Sigma,\calD,n} &:= \frac{\langle X_{1,\calD}^n+X_{2,\calD}^n , Z_\calD^n\rangle - \frac{P_\Sigma}{2} (\|Z_\calD^n\|^2-d_n)}{q_\Sigma\sqrt{n}} .
\end{align}
There exists output distributions $Q_{Y^n|X_2}$ and $Q_{Y^n}$ for the two Verd\'u--Han constraints (that result in information densities $\iota_{1,n}$ and $\iota_{\Sigma,n}$)  and nonnegative random variables $R_{1,n}$ and $R_{\Sigma,n}$ with $( R_{1,n}, R_{\Sigma,n})\xrightarrow{\rmP}(0,0)$ such that 
    \begin{align}
        \frac{\iota_{1,n}-n\rvC(P_1)}{\sqrt{n}} &\le A_{1,\calD,n} + R_{1,n} , \label{eqn:id_ub}\\
        \frac{\iota_{\Sigma,n}-n\rvC(P_\Sigma)}{\sqrt{n}} &\le A_{\Sigma,\calD,n} + \frac{T_{\calD,n}}{q_\Sigma}+ R_{\Sigma,n} . \label{eqn:id_ub_sum}
    \end{align}
The output distributions are given explicitly in \eqref{eqn:out1} and \eqref{eqn:out_sum}.
\end{proposition}
\begin{proof}
    By the energy expurgation operation in~\eqref{eqn:energy_exp}, one has, almost surely,
    \begin{align}
        \|X_{1,\calE}^n\|^2\le b_n\quad\mbox{and}\quad\|X_{2,\calE}^n\|^2\le b_n. \label{eqn:bound_bn}
    \end{align}
    Furthermore,
    \begin{align}
        \|X_{1,\calE}^n+ X_{2,\calE}^n\|^2\le (  \|X_{1,\calE}^n \| + \|X_{2,\calE}^n \|)^2 \le (\sqrt{b_n}+\sqrt{b_n})^2=4b_n. \label{eqn:bound_4bn}
    \end{align}
    These are precisely the two deterministic bounds used to choose the broad Gaussian variances below; the factor 4 is given by the use of the triangle inequality in \eqref{eqn:bound_4bn}. 

    As usual, set $Z_\calE^n=\Pi_{\calE_n}Z^n$. If $r_n=0$, define $\calL_{1,\calE,n}=\calL_{\Sigma,\calE,n}=0$ and omit all $\calE_n$ factors. For $r_n\ge1$, define the output distributions
    \begin{align}
        Q_{Y^n|X_2^n}(\cdot|x_2^n) &:= \calN_{\calD_n}(x_{2,\calD}^n, q_1 I_{\calD_n})\otimes\calN_{\calE_n}\bigg(x_{2,\calE}^n, \Big(1+\frac{b_n}{r_n} \Big)I_{\calE_n}\bigg), \label{eqn:out1} \\
        Q_{Y^n}  &:= \calN_{\calD_n}(0, q_1 I_{\calD_n})\otimes\calN_{\calE_n}\bigg(0, \Big(1+\frac{4b_n}{r_n}\Big) I_{\calE_n}\bigg),\label{eqn:out_sum} 
    \end{align}
    and define 
    \begin{align}
       \calL_{1,\calE,n} =\Lambda_{b_n}(X_{1,\calE}^n, Z_{\calE}^n)\quad \mbox{and}\quad \calL_{\Sigma,\calE,n} =\Lambda_{4b_n}(X_{1,\calE}^n+X_{2,\calE}^n, Z_{\calE}^n).
    \end{align}
    When $r_n=0$ the desired convergence is immediate from the zero convention. For indices with $r_n\ge1$, to pass explicitly from the deterministic vector statement \eqref{eqn:prop_cheby} to these random codewords, fix $\eta>0$ and condition on the messages. Since $Z_\calE^n$ is independent of both codewords,  \eqref{eqn:bound_bn} gives
    \begin{align}
        \Pr\big(  \calL_{1,\calE,n}^+ >\eta\sqrt{n}\big) &=\bbE \Big[ \Pr\big(  \Lambda_{b_n}(X_{1,\calE}^n, Z_{\calE}^n)^+ >\eta\sqrt{n}\mid X_{1,\calE}^n\big) \Big] \nn\\
        &\le \sup_{ x_{1,\calE}^n:\| x_{1,\calE}^n \|^2\le b_n}  \Pr\big(  \Lambda_{b_n}(x_{1,\calE}^n, Z_{\calE}^n)^+ >\eta\sqrt{n}\big)\to0. \label{eqn:apply_cheby}
    \end{align}
    Similarly,  $\Pr\big(  \calL_{\Sigma,\calE,n}^+ >\eta\sqrt{n}\big)\to0$. Thus, 
    \begin{align}
        \frac{\calL_{1,\calE,n}^+ }{\sqrt{n}}\xrightarrow{\rmP}0\quad \mbox{and}\quad \frac{\calL_{\Sigma,\calE,n}^+ }{\sqrt{n}}\xrightarrow{\rmP}0. \label{eqn:Lplus}
    \end{align}

Define $\kappa(P):=\rvC(P)-\frac{P}{2(P+1)}$. Since $\kappa(0)=0 $ and $\kappa'(P)>0$, one has $\kappa(P)>0$ for all $P>0$. Under the final message-uniform  laws, put 
\begin{align}
    a_j^{(n)} := \| X_{j,\calE}^n \|^2,\qquad j \in \{1,2\}.
\end{align}
By \eqref{eqn:bound_bn}, $a_j^{(n)}\le b_n$ almost surely for $j \in \{1,2\}$.

Recall that $d_n=\mathrm{dim}(\calD_n)=n-r_n$. The diffuse individual log-likelihood is 
\begin{align}
    \iota_{1, \calD, n}= d_n \rvC(P_1)+ \frac{\| X_{1,\calD}^n\|^2- d_nP_1}{2q_1} + \frac{\langle X_{1,\calD}^n, Z_{\calD}^n\rangle-\frac{P_1}{2} ( \| Z_{\calD}^n\|^2-d_n) }{q_1}
\end{align}
The product form of the output distribution gives $\iota_{1,n}=\iota_{1, \calD, n}+\calL_{1,\calE,n}$. Because $\|X_{1,\calD}^n\|^2=nP_1-a_1^{(n)}$ and $d_n=n-r_n$,
\begin{align}
    \iota_{1,n}-n\rvC(P_1)=-r_n\kappa(P_1) -\frac{a_1^{(n)}}{2q_1}+\frac{\langle X_{1,\calD}^n, Z_{\calD}^n\rangle-\frac{P_1}{2} ( \| Z_{\calD}^n\|^2-d_n) }{q_1} +\calL_{1,\calE,n}.\label{eqn:info_dens_1}
\end{align}
Similarly, $\|X_{1,\calD}^n+X_{2,\calD}^n\|^2=nP_\Sigma-a_1^{(n)}- a_2^{(n)}-2 \langle X_{1,\calD}^n, X_{2,\calD}^n \rangle$, and hence, the diffuse joint log-likelihood is
\begin{align}
    \iota_{\Sigma,n}-n\rvC(P_\Sigma) &=-r_n\kappa(P_\Sigma) -\frac{a_1^{(n)}+a_2^{(n)}}{2q_\Sigma}+ \frac{\langle X_{1,\calD}^n+X_{2,\calD}^n, Z_{\calD}^n\rangle-\frac{P_\Sigma}{2} ( \| Z_{\calD}^n\|^2-d_n) }{q_\Sigma} \nn\\
    &\qquad+\frac{ \langle X_{1,\calD}^n, X_{2,\calD}^n \rangle}{q_\Sigma} +\calL_{\Sigma,\calE,n}. \label{eqn:info_dens_sigma}
\end{align}
Importantly, the exceptional overlap does not appear separately in \eqref{eqn:info_dens_sigma}; it is contained in $\calL_{\Sigma,\calE,n}$, which is bounded using Chebyshev in \eqref{eqn:apply_cheby}.

The first two terms on the right-hand sides of~\eqref{eqn:info_dens_1} and \eqref{eqn:info_dens_sigma} are non-positive for every codeword pair. Hence, 
\begin{align}
    \frac{\iota_{1,n}-n\rvC(P_1)}{\sqrt{n}} &\le A_{1,\calD,n } +\frac{\calL_{1,\calE,n}^+}{\sqrt{n}},\\
    \frac{\iota_{\Sigma,n}-n\rvC(P_\Sigma)}{\sqrt{n}} &\le A_{\Sigma,\calD,n } +\frac{ T_{\calD,n}}{q_\Sigma} + \frac{\calL_{\Sigma,\calE,n}^+}{\sqrt{n}}.
\end{align}
Thus, take $R_{1,n} := \calL_{1,\calE,n}^+/\sqrt{n}$ and  $R_{\Sigma,n} := \calL_{\Sigma,\calE,n}^+/\sqrt{n}$, which are sequences of non-negative random variables that converge to zero in probability by~\eqref{eqn:Lplus}.  This proves \eqref{eqn:id_ub} and \eqref{eqn:id_ub_sum}.
\end{proof}
\subsection{Conditional characteristic functions and the joint Gaussian limit} \label{sec:char_functions}
Proposition \ref{prop:info_dens} reduces the converse to the two diffuse Gaussian-noise statistics and the diffuse overlap. In the following proposition, we show that $(A_{1,\calD,n }, A_{\Sigma,\calD,n },T_{\calD,n})$  in~\eqref{eqn:id_ub} and \eqref{eqn:id_ub_sum} satisfies a joint Gaussian law asymptotically. We also prove asymptotic independence of the overlap $T_{\calD,n}$ from the noise terms $(A_{1,\calD,n }, A_{\Sigma,\calD,n })$.
\begin{proposition}[Joint Gaussian limit and asymptotic independence] \label{prop:clt}
    Let $(A_1, A_\Sigma)$ be a zero-mean bivariate Gaussian random vector with covariance matrix 
    \begin{align}
        \begin{bmatrix}
        \rvV(P_1) & V_{1,\Sigma}\\
V_{1,\Sigma} & \rvV(P_\Sigma) 
    \end{bmatrix} \label{eqn:Vcov}
    \end{align}
    and let $T\sim\calN(0,P_1P_2)$ be independent of $(A_1, A_\Sigma)$. Then, 
    \begin{align}
        (A_{1,\calD,n }, A_{\Sigma,\calD,n },T_{\calD,n})\Rightarrow(A_1, A_\Sigma, T). \label{eqn:AAT_conv}
    \end{align}
    Consequently, 
    \begin{align}
        \bigg( A_{1,\calD,n }, A_{\Sigma,\calD,n }+\frac{T_{\calD,n}}{q_\Sigma}\bigg)\Rightarrow \bigg(A_1, A_\Sigma+\frac{T}{q_\Sigma} \bigg)\sim\calN(0, \bV_{\mathrm{fc}}), \label{eqn:proof_fc}
    \end{align}
    where $\bV_{\mathrm{fc}}$ is the first-corner  dispersion matrix defined in \eqref{eqn:defVfc}.
\end{proposition}
\begin{proof}
    Throughout this proof, we use $\rmi$ to denote the imaginary unit $\rmi :=\sqrt{-1}$. Let $\calF_n := \sigma(X_{1,\calD}^n, X_{2,\calD}^n)$.

    For fixed $\theta_1, \theta_\Sigma\in\bbR$, define 
    \begin{align}
        \xi^{(n)}:=\underbrace{\left(\frac{\theta_1}{ q_1}+ \frac{\theta_\Sigma}{ q_\Sigma}\right)}_{ =:u} X_{1,\calD}^n+\underbrace{\frac{\theta_\Sigma}{ q_\Sigma}}_{ =:v}    X_{2,\calD}^n\quad \mbox{and}\quad \chi := \frac{\theta_1 P_1}{ q_1}+ \frac{\theta_\Sigma P_\Sigma}{ q_\Sigma}. \label{eqn:xi_def}
    \end{align}
    Then 
    \begin{align}
        \theta_1 A_{1,\calD,n}+\theta_\Sigma A_{\Sigma,\calD,n} = \frac{\langle\xi^{(n)},Z_\calD^n\rangle }{\sqrt{n}}- \frac{\chi}{2\sqrt{n}} ( \|Z_\calD^n\|^2-d_n).
    \end{align}
    Gaussian integration yields the characteristic function 
    \begin{align}
        \Psi_n(\theta_1,\theta_\Sigma) &:=\bbE\big[ e^{\rmi (\theta_1 A_{1,\calD,n}+\theta_\Sigma A_{\Sigma,\calD,n}  )} \mid\calF_n\big] \nn\\
        &=e^{\rmi \chi d_n/(2\sqrt{n}}\left( 1+\frac{\rmi\chi}{\sqrt{n}}\right)^{-d_n/2}\exp\bigg[ - \frac{\|\xi^{(n)} \|^2}{2n(1+\rmi\chi/\sqrt{n})} \bigg].
    \end{align}
    Now, by~\eqref{eqn:bound_rn}, $d_n/n\to 1$  and by~\eqref{eqn:energyD}, $\| X_{j,\calD}^n \|^2/n\to P_j$ uniformly for $j \in \{1,2\}$.  Since \(T_{\calD,n}\) converges in distribution, the sequence 
\((T_{\calD,n})_{n\ge1}\) is tight~\cite{Kallenberg2021}, and therefore
    \begin{align}
        \frac{ \langle X_{1,\calD}^n, X_{2,\calD}^n\rangle}{n}=\frac{ T_{\calD,n}}{\sqrt{n}}\xrightarrow{\rmP}0. \label{eqn:tight}
    \end{align}
    From~\eqref{eqn:xi_def}, the following expansion holds
    \begin{align}
        \frac{\|\xi^{(n)}\|^2}{n}=u^2 \frac{\|X_{1,\calD}^n\|^2}{n}+ v^2 \frac{\|X_{2,\calD}^n\|^2}{n}+ 2uv \frac{ \langle X_{1,\calD}^n, X_{2,\calD}^n\rangle}{n}.
    \end{align}
    The first two   normalized squared norms converge (uniformly) to $P_1$ and $P_2$ respectively, while the last term converges in probability to zero by \eqref{eqn:tight}. Because $u$ and $v$ are fixed real numbers, 
    \begin{align}
        \frac{\|\xi^{(n)}\|^2}{n}\xrightarrow{\rmP}P_1u^2+P_2v^2.
    \end{align}
    Also, 
    \begin{align}
        \frac{\rmi \chi d_n}{2\sqrt{n}}-\frac{d_n}{2}\log\left(1+\frac{\rmi\chi}{\sqrt{n}} \right)= -\frac{\chi^2 d_n}{4n}+o(1)\to -\frac{\chi^2}{4}.
    \end{align}
    Consequently, 
    \begin{align}
         \Psi_n(\theta_1,\theta_\Sigma)\xrightarrow{\rmP} \Psi(\theta_1,\theta_\Sigma) = \exp\left\{ -\frac{1}{2}\bigg[P_1 u^2 + P_2v^2 + \frac{\chi^2}{2}\bigg] \right\}. \label{eqn:Psi_def}
    \end{align}
Since \(T_{\calD,n}\) is \(\calF_n\)-measurable, the tower property gives,
for every \(s,\theta_1,\theta_\Sigma\in\bbR\),
\begin{align}
 \bbE\left[
e^{ 
\rmi sT_{\calD,n}
+\rmi\theta_1 A_{1,n}
+\rmi\theta_\Sigma A_{\Sigma,n}}
\right]  & =
\bbE\left[
\bbE \big[
e^{ \rmi sT_{\calD,n}
+\rmi\theta_1 A_{1,n}
+\rmi\theta_\Sigma A_{\Sigma,n}
 }
\mid\calF_n
 \big]
\right] \nn\\
&=
\bbE\left[
e^{\rmi sT_{\calD,n}}
\bbE \big[
e^{ 
\rmi\theta_1 A_{1,n}
+\rmi\theta_\Sigma A_{\Sigma,n}
}
\mid\calF_n
 \big]
\right]  =
\bbE\left[
e^{\rmi sT_{\calD,n}}
\Psi_n(\theta_1,\theta_\Sigma)
\right].
\label{eqn:tower-cf}
\end{align}
% The second equality is where the
% \(\calF_n\)-measurability of \(T_{\calD,n}\) is used.

Moreover, since conditional characteristic functions have modulus at most
one, the convergence
\(\Psi_n(\theta_1,\theta_\Sigma)\to
\Psi(\theta_1,\theta_\Sigma)\) also holds in \(L^1\).
Consequently,
\begin{align}
 \left|
\bbE\left[
e^{\rmi sT_{\calD,n}}
\big(
\Psi_n(\theta_1,\theta_\Sigma)
-\Psi(\theta_1,\theta_\Sigma)
\big)
\right]
\right|  \le
\bbE\left[
\left|
\Psi_n(\theta_1,\theta_\Sigma)
-\Psi(\theta_1,\theta_\Sigma)
\right|
\right]
\to 0,
\end{align}
where the inequality uses only
\(\lvert e^{\rmi sT_{\calD,n}}\rvert=1\).
    % The convergence also holds in $L^1$ because $|\Psi_n|\le 1$. Since $T_{\calD,n}$ is $\calF_n$-measurable, for every $s\in\bbR$, 
    % \begin{align}
    %     \left| \bbE\left[ e^{\rmi s T_{\calD,n}}  \big(  \Psi_n(\theta_1,\theta_\Sigma)-\Psi(\theta_1,\theta_\Sigma)  \big) \right]\right|\le \bbE \left[ \big|  \Psi_n(\theta_1,\theta_\Sigma)-\Psi(\theta_1,\theta_\Sigma)  \big| \right]\to 0.
    % \end{align}
    Combining this with the fact that $T_{\calD,n}\Rightarrow\calN(0,P_1P_2)$ gives
    \begin{align}
        \bbE\left[e^{\rmi ( s T_{\calD,n} +\theta_1 A_{1,\calD,n}+\theta_\Sigma A_{\Sigma,\calD,n}   )} \right]=\bbE\left[ e^{\rmi s T_{\calD,n}}\Psi_n(\theta_1,\theta_\Sigma) \right]\to e^{-s^2 P_1P_2/2}\Psi(\theta_1,\theta_\Sigma). \label{eqn:cf_limit}
    \end{align}
    From the quadratic form in the exponent of $\Psi$ in~\eqref{eqn:Psi_def},  we see that it is the characteristic function of the zero mean Gaussian random vector $(A_1,A_\Sigma)$. Moreover $e^{-s^2P_1P_2/2}$ is the characteristic function of the independent random variable $T$ specified there. Thus, the limit in \eqref{eqn:cf_limit} is the joint characteristic function of  $(A_1,A_\Sigma,T)$, and L\'evy's continuity theorem proves~\eqref{eqn:AAT_conv}. 

    Expanding the quadratic form in~\eqref{eqn:Psi_def} gives
    \begin{align}
        \mathrm{Var}(A_1)  &= \frac{P_1+P_1^2}{q_1^2}=\rvV(P_1), \\
       \mathrm{Var}(A_\Sigma)  &= \frac{P_\Sigma+P_\Sigma^2}{q_\Sigma^2}=\rvV(P_\Sigma) ,\\ 
       \mathrm{Cov}(A_1, A_\Sigma) & = \frac{P_1 + P_1P_\Sigma/2}{q_1 q_\Sigma} = V_{1,\Sigma}.
    \end{align}
    Adding the independent $T/q_\Sigma$ to the second coordinate gives 
    \begin{align}
           \mathrm{Var}\bigg(A_\Sigma+\frac{T}{q_\Sigma} \bigg) =\rvV(P_\Sigma) +\frac{P_1P_2}{q_\Sigma^2}.
    \end{align}
    This completes the proof of~\eqref{eqn:proof_fc}.
\end{proof}
\subsection{Completion of the proof} \label{sec:completion}
Proposition \ref{prop:info_dens} gives component-wise upper bounds for the
two information densities and Proposition~\ref{prop:clt} gives the joint limit of their diffuse
upper bounds. It remains to insert the final subcode sizes into the Verd\'u--Han inequality stated in Lemma~\ref{lem:vh}.

\begin{proof}[Proof of Theorem~\ref{thm:second-order_outer}]
Let $M_{j,n}^\rmf:= |\calM_{j,n}^\rmf|$ for $j \in \{1,2\}$ be the sizes of the final message sets. For the final subcode, retain the original correct-decision functions $d_{m_1,m_2}$ only for $(m_1, m_2) \in \calM_{1,n}^\rmf\times \calM_{2,n}^\rmf$. Clearly, they remain subnormalized for each $y^n$, i.e., 
\begin{align}
    \sum_{ m_1 \in \calM_{1,n}^\rmf} \sum_{ m_2\in \calM_{2,n}^\rmf} d_{m_1,m_2}(y^n)\le 1. 
\end{align}
The decoder outputs outside the retained rectangle are treated   as errors. With
this convention, the average error is  $\eps_n'$ as defined in Proposition~\ref{prop:final}. We now apply the Verd\'u--Han inequality to the final product subcode and the hybrid output distributions in \eqref{eqn:out1} and \eqref{eqn:out_sum}. With $\gamma=\gamma_n=n^{1/4}$,  we obtain
\begin{align}
 1-\eps_n'\le\Pr \left(\begin{bmatrix}
\iota_{1,n} (X_1^n,X_2^n,Y^n) \\ \iota_{\Sigma,n} (X_1^n,X_2^n,Y^n)
\end{bmatrix} \ge \begin{bmatrix}
\log M_{1,n}^\rmf -\gamma_n\\ \log (M_{1,n}^\rmf M_{2,n}^\rmf) -\gamma_n
\end{bmatrix} \right) + 2e^{-\gamma_n}, \label{eqn:vh_final}
\end{align}
Fix $\eta>0$. For large enough $n$, the final message sizes satisfy
\begin{align}
    \log M_{1,n}^\rmf
&\ge n\rvC(P_1)+(L_1-\eta)\sqrt n, \label{eqn:final_standing-rate1}\\
\log (M_{1,n}^\rmf M_{2,n}^\rmf)
&\ge n \rvC(P_\Sigma) 
 +(L_1+L_2-\eta)\sqrt n.\label{eqn:final_standing-rate2}
\end{align}
By Proposition~\ref{prop:info_dens}, with probability $1-o(1)$, 
\begin{align}
    R_{1,n} \le \eta\quad\mbox{and}\quad R_{\Sigma,n} \le \eta.
\end{align}
%For all sufficiently large $n$, the $o(1) $ terms in \eqref{eqn:final_standing-rate1} and \eqref{eqn:final_standing-rate2} can be upper bounded by $\eta$. 
The component-wise bounds in \eqref{eqn:id_ub} and \eqref{eqn:id_ub_sum} imply that, outside an event of probability $o(1)$, the event in \eqref{eqn:vh_final}, denoted as $\calV_n$, is contained in the orthant
\begin{align}
    \calO_{\eta,n}:=\left\{ A_{1,\calD,n}\ge L_1-2\eta, ~ A_{\Sigma,\calD,n} + \frac{T_{\calD,n}}{q_\Sigma}\ge L_1+L_2-2\eta\right\}.
\end{align}
Equivalently, $\Pr(\calV_n)\le\Pr(\calO_{\eta,n})+o(1)$. Let $(Z_1, Z_\Sigma)\sim\calN(0,\bV_{\mathrm{fc}})$, where $\bV_{\mathrm{fc}} $ is the first-corner dispersion matrix defined in~\eqref{eqn:defVfc}. Using  the CLT result in~\eqref{eqn:proof_fc} yields
\begin{align}
    \limsup_{n\to\infty}\Pr(\calV_n)\le \Pr( Z_1\ge L_1-2\eta,~Z_\Sigma\ge L_1+L_2-2\eta). 
\end{align}
Letting $\eta\to 0^+$ and using the continuity of the bivariate Gaussian CDF yields
\begin{align}
    \limsup_{n\to\infty}\Pr(\calV_n)\le \Pr( Z_1\ge L_1,~Z_\Sigma\ge L_1+L_2). 
\end{align}
Since $2e^{-\gamma_n}=2e^{-n^{1/4}}\to0$, we have 
\begin{align}
    1-\eps \le\liminf_{n\to\infty} (1-\eps_n')\le\limsup_{n\to\infty} \Pr( \calV_n)
     \le  \Pr( Z_1\ge L_1,~Z_\Sigma\ge L_1+L_2).
\end{align}
This is equivalent to \eqref{eqn:outer_bd}.
\end{proof}
\section{Conclusion and Future Work} \label{sec:concl}
This work establishes the second-order coding region of the two-user Gaussian MAC at the corner points of its (pentagonal) capacity region, proving that the existing second-order  achievability bounds in~\cite{yavas21,MolavianJaziLaneman2015,ScarlettMartinezGuillen2015} are tight at the corner points. At a conceptual level, the converse translates near-capacity rate constraints into geometric regularity of the codebooks while preserving their product structure. Rectangular trimming and output-distribution regularity produce two log-determinant deficits, which induce the diffuse-exceptional spectral decomposition. The diffuse component recovers the Gaussian fluctuation of the codeword inner product appearing in the achievability analysis, whereas the exceptional component is shown to be negligible on the $\sqrt n$ scale. For the corner-point problem, this framework provides an alternative to the pairwise expurgation and wringing methods underlying the classical MAC strong converses~\cite{Dueck1981StrongConverse,Ahlswede1982Elementary,fongtan16, Kosut2022Wringing, WeiKosut2021Wringing}.

Related  constructions of hybrid output distributions have appeared in sharp point-to-point converses. In Hayashi's discrete memoryless channel converse~\cite{Hayashi09}, input sequences are partitioned according to their empirical distributions $\hatP$ into a near-capacity class, satisfying $I(\hatP,\rvW)\ge \rvC-\eta$, and its complement. The former is analyzed relative to the capacity-achieving  output distribution, whereas the latter is controlled using the output distribution induced by the corresponding input type. For the continuous-time Poisson channel, Sakai, Tan, and Kova\v{c}evi\'c~\cite{STK2020} construct a single hybrid  output distribution containing boundary product measures and a weighted mixture of product output laws indexed by a fine grid of input distributions. Input types separated from the capacity-achieving type are controlled using the boundary components, while nearby types are matched to an appropriate component of the mixture to obtain the sharp normal approximation. These arguments share a broad principle with the present diffuse-exceptional decomposition: a refined second-order analysis is applied in the regular regime, while the complementary regime is handled by a coarser bound adapted to its structure, typically using a Chebyshev argument similarly to Proposition~\ref{prop:excep_likeli}.

A natural open question is whether the exact second-order region can also be determined in the relative interior of the sum-rate face, namely at rate pairs $(R_1^*, R_2^*)$ satisfying
\begin{equation}
R_1^*+R_2^*=\rvC(P_\Sigma),
\qquad \rvC\left(\frac{P_1}{1+P_2}\right)<R_1^*<\rvC(P_1).
\end{equation} 
The current techniques, however, do not carry over to this case. The obstruction is   the loss of the individual-rate  hypothesis~\eqref{eqn:standing-rate1} used in Proposition \ref{prop:trim1}. At the first corner, $\log M_{1,n} = n\rvC(P_1) + O(\sqrt{n})$, so the Polyanskiy--Verd\'u theorem~\cite{PolVerdu14}   gives $\Delta_{1,n}=O(\sqrt{n})$ (cf.\ Proposition~\ref{prop:trim1}(iv)). In the relative interior of the sum-rate face, however, $R_1^*< \rvC(P_1)$, and the same theorem only gives $\Delta_{1,n}=O(n)$. Consequently, Proposition~\ref{prop:user2} transfers only a success probability as small as $\exp[-O(n)]$, and the vanishing success covariance converse (Proposition~\ref{prop:excep_likeli}) yields at best $\Gamma_{2,n}=O(n),$ rather than $\Gamma_{2,n}=O(n^{3/4})$. The resulting spectral estimates are then too weak to make the exceptional contribution negligible on the $\sqrt n$ scale, so the Brascamp--Lieb argument and the subsequent Gaussian-limit calculation no longer work. Applying output regularity directly to the sum codebook controls $X_1^n+X_2^n$, but does not separately control the two codebooks or their inner product. Resolving the interior of the sum-rate face therefore appears to require a new sum-rate-specific regularity principle that preserves the product structure while recovering such separate geometric control. This, we leave to future work.

\section*{Acknowledgments}
\emph{Generative-AI use disclosure:}
The author used OpenAI Codex (GPT-5.6 Sol, Ultra mode) during the development and preparation of this
manuscript. The system assisted in exploring and stress-testing proof
strategies, formulating and refining intermediate lemmas, checking
algebraic and probabilistic calculations and possible logical gaps, and
refining portions of the exposition. Every mathematical statement, proof, and reference appearing in
the final manuscript was independently verified by the author, who
assumes full responsibility for the contents of the manuscript.

The author thanks Recep Can Yavas for encouraging him to explore the use of
OpenAI's GPT-5.6 Sol model (Ultra mode) in investigating
information-theoretic problems. The author also thanks Jonathan Scarlett and
Lin Zhou for constructive comments on the manuscript.
\bibliographystyle{ieeetr}
\bibliography{isitbib}

@STRING{ISIT = {Proc. Int. Symp. Inform. Theory}}

@STRING{TIT = {IEEE Trans. Inform. Theory}}

@article{Augustin1966,
  author  = {Augustin, U.},
  title   = {Ged{\"a}chtnisfreie Kan{\"a}le f{\"u}r diskrete Zeit},
  journal = {Zeitschrift f{\"u}r Wahrscheinlichkeitstheorie und Verwandte Gebiete},
  volume  = {6},
  pages   = {10--61},
  year    = {1966}
}

@inproceedings{LiuCourtadeCuffVerdu2016,
  author    = {J. Liu and T. A. Courtade  and P.
               Cuff and S. Verd{\'u}},
  title     = {{Brascamp--Lieb} Inequality and Its Reverse:
               An Information-Theoretic View},
  booktitle = {Proceedings of the 2016 IEEE International Symposium
               on Information Theory (ISIT)},
  pages     = {1048--1052},
  year      = {2016},
  address = "Barcelona, Spain",
  month     = jul,
  publisher = {IEEE},
  doi       = {10.1109/ISIT.2016.7541459}
}

@article{CarlenCorderoErausquin2009,
  author  = {E. A.  Carlen and D. Cordero-Erausquin},
  title   = {Subadditivity of the Entropy and Its Relation to
             {Brascamp--Lieb} Type Inequalities},
  journal = {Geometric and Functional Analysis},
  volume  = {19},
  number  = {2},
  pages   = {373--405},
  year    = {2009},
  doi     = {10.1007/s00039-009-0001-y}
}

@book{Bhatia1997,
  author    = {Bhatia, Rajendra},
  title     = {Matrix Analysis},
  series    = {Graduate Texts in Mathematics},
  volume    = {169},
  publisher = {Springer},
  address   = {New York},
  year      = {1997}
}

@incollection{Cover1975SomeAdvances,
  author    = {Cover, Thomas M.},
  title     = {Some Advances in Broadcast Channels},
  editor    = {Viterbi, Andrew J.},
  booktitle = {Advances in Communication Systems: Theory and Applications},
  volume    = {4},
  pages     = {229--260},
  year      = {1975},
  publisher = {Academic Press},
  address   = {San Francisco, CA},
  doi       = {10.1016/B978-0-12-010904-3.50011-7}
}

@article{MolavianJaziLaneman2015,
  author  = {E. MolavianJazi  and J. N. Laneman},
  title   = {A Second-Order Achievable Rate Region for {Gaussian} Multi-Access Channels via a Central Limit Theorem for Functions},
  journal = TIT,
  volume  = {61},
  number  = {12},
  pages   = {6719--6733},
  month   = dec,
  year    = {2015},
  doi     = {10.1109/TIT.2015.2492547}
}

@article{ScarlettMartinezGuillen2015,
  author  = {J. Scarlett and  A.  Martinez and A. {Guill{\'e}n i F{\`a}bregas}},
  title   = {Second-Order Rate Region of Constant-Composition Codes
             for the Multiple-Access Channel},
  journal = TIT,
  volume  = {61},
  number  = {1},
  pages   = {157--172},
  month   = jan,
  year    = {2015},
  doi     = {10.1109/TIT.2014.2371026}
}

@book{DupuisEllis1997,
  author    = {Dupuis, Paul and Ellis, Richard S.},
  title     = {A Weak Convergence Approach to the Theory of
               Large Deviations},
  series    = {Wiley Series in Probability and Statistics},
  publisher = {John Wiley \& Sons},
  address   = {New York},
  year      = {1997},
  doi       = {10.1002/9781118165904}
}

@article{DonskerVaradhan1975,
  author  = {Donsker, M. D. and Varadhan, S. R. S.},
  title   = {Asymptotic Evaluation of Certain {Markov} Process
             Expectations for Large Time, {I}},
  journal = {Communications on Pure and Applied Mathematics},
  volume  = {28},
  number  = {1},
  pages   = {1--47},
  month   = jan,
  year    = {1975},
  doi     = {10.1002/cpa.3160280102}
}

@inproceedings{WangColbeckRenner2009,
  author    = {L. Wang and R. Colbeck and R. Renner},
  title     = {Simple Channel Coding Bounds},
  booktitle = {Proceedings of the 2009 IEEE International Symposium on
               Information Theory (ISIT)},
  address   = {Seoul, South Korea},
  pages     = {1804--1808},
  year      = {2009},
  doi       = {10.1109/ISIT.2009.5205312}
}

@ARTICLE{STK2020,
  author = {Y.  Sakai and V. Y. F. Tan and M. Kova\v{c}evi\'{c}},
  title = {Second- and Third-Order Asymptotics of the Continuous-Time {Poisson} Channel},
  journal = TIT,
  year = {2020},
  volume = {66},
  pages = {4742--4760},
  number = {8},
  month = {Aug}
}

@ARTICLE{ST2015,
  author = {J.  Scarlett and V. Y. F. Tan},
  title = {Second-Order Asymptotics for the {Gaussian  MAC} With Degraded Message Sets},
  journal = TIT,
  year = {2015},
  volume = {61},
  pages = {6700--6718},
  number = {12},
  month = {Dec}
}

@ARTICLE{yavas21,
  author = {R. C. Yavas and V. Kostina and M. Effros},
  title = {Gaussian Multiple and Random Access Channels: Finite-Blocklength Analysis},
  journal = TIT,
  year = {2021},
  volume = {67},
  pages = {6983--7009},
  number = {11},
  month = {Nov}
}

@article{SoleBelfiore2013,
  author  = {Patrick Sol{\'e} and Jean-Claude Belfiore},
  title   = {Constructive Spherical Codes Near the {Shannon} Bound},
  journal = {Designs, Codes and Cryptography},
  volume  = {66},
  number  = {1--3},
  pages   = {17--26},
  year    = {2013},
  doi     = {10.1007/s10623-012-9633-2}
}

@book{ConwaySloane1999,
  author    = {John H. Conway and Neil J. A. Sloane},
  title     = {Sphere Packings, Lattices and Groups},
  edition   = {3},
  series    = {Grundlehren der mathematischen Wissenschaften},
  volume    = {290},
  publisher = {Springer},
  address   = {New York},
  year      = {1999}
}

@ARTICLE{PolVerdu14,
  author = {Y. Polyanskiy and S. Verd\'u},
  title = {Empirical Distribution of Good Channel Codes With Nonvanishing Error Probability},
  journal = TIT,
  year = {2014},
  volume = {60},
  pages = {5--21},
  number = {1},
  month = {Jan}
}

@ARTICLE{fongtan16,
  author = {S. L. Fong and V. Y. F. Tan},
  title = {A Proof of the Strong Converse Theorem for {Gaussian} Multiple Access Channels},
  journal = TIT,
  year = {2016},
  volume = {62},
  pages = {4376--4394},
  number = {8},
  month = {Aug}
}

@ARTICLE{wyner74,
  author = {A. D. Wyner},
  title = {{Recent results in the Shannon theory}},
  journal = TIT,
  year = {1974},
  volume = {20},
  pages = {2--10},
  number = {1},
  month = {Jan}
}

@BOOK{elgamal,
  title = {Network Information Theory},
  publisher = {Cambridge University Press},
  year = {2012},
  author = {A. {El~Gamal} and Y.-H. Kim},
  address = {Cambridge, U.K.}
}

@BOOK{Han10,
  title = {Information-Spectrum Methods in Information Theory},
  publisher = {Springer Berlin Heidelberg},
  year = {2003},
  author = {T. S. Han},
  month = {Feb}
}

@ARTICLE{Han98,
  author = {T. S. Han},
  title = {An Information-Spectrum Approach to Capacity Theorems for the General 	Multiple-Access Channel},
  journal = TIT,
  year = {1998},
  volume = {44},
  pages = {2773--2795},
  number = {7},
  month = {Jul}
}

@book{Kallenberg2021,
  author    = {O. Kallenberg},
  title     = {Foundations of Modern Probability},
  edition   = {3},
  publisher = {Springer},
  address   = {Cham},
  year      = {2021},
  doi       = {10.1007/978-3-030-61871-1}
}

@inproceedings{WeiKosut2021Wringing,
  author    = {F. Wei and O. Kosut},
  title     = {A Wringing-Based Proof of a Second-Order Converse for the
               Multiple-Access Channel under Maximal Error Probability},
  booktitle = {2021 IEEE International Symposium on Information Theory (ISIT)},
  address = "Melbourne, Australia",
  pages     = {2220--2225},
  year      = {2021},
  doi       = {10.1109/ISIT45174.2021.9518136}
}

@article{Kosut2022Wringing,
  author  = {O. Kosut},
  title   = {A Second-Order Converse Bound for the Multiple-Access
             Channel via Wringing Dependence},
  journal = TIT,
  volume  = {68},
  number  = {6},
  pages   = {3552--3584},
  month   = jun,
  year    = {2022},
  doi     = {10.1109/TIT.2022.3151711},
  url     = {https://arxiv.org/abs/2007.15664}
}

@article{Margulis1974Probabilistic,
  author  = {Margulis, G. A.},
  title   = {Probabilistic Characteristics of Graphs with Large
             Connectivity},
  journal = {Problemy Peredachi Informatsii},
  volume  = {10},
  number  = {2},
  pages   = {101--108},
  year    = {1974},
  note    = {English translation in Problems of Information Transmission,
             vol.~10, no.~2, pp.~174--179},
  url     = {https://www.mathnet.ru/eng/ppi1034}
}

@article{AhlswedeGacsKorner1976,
  author  = {R. Ahlswede and P. G{\'a}cs and J.  K{\"o}rner},
  title   = {Bounds on Conditional Probabilities with Applications in
             Multi-User Communication},
  journal = {Zeitschrift f{\"u}r Wahrscheinlichkeitstheorie und
             Verwandte Gebiete},
  volume  = {34},
  number  = {2},
  pages   = {157--177},
  year    = {1976},
  doi     = {10.1007/BF00535682}
}

@article{Ahlswede1982Elementary,
  author  = {R. Ahlswede},
  title   = {An Elementary Proof of the Strong Converse Theorem for the
             Multiple-Access Channel},
  journal = {Journal of Combinatorics, Information and System Sciences},
  volume  = {7},
  number  = {3},
  pages   = {216--230},
  year    = {1982}
}

@article{Dueck1981StrongConverse,
  author  = {G. Dueck},
  title   = {The Strong Converse to the Coding Theorem for the
             Multiple-Access Channel},
  journal = {Journal of Combinatorics, Information and System Sciences},
  volume  = {6},
  number  = {3},
  pages   = {187--196},
  year    = {1981}
}

@ARTICLE{TanTom15,
  author = {V. Y. F. Tan and M. Tomamichel},
  title = {The Third-Order Term in the Normal Approximation for the {AWGN} Channel},
  journal = TIT,
  year = {2015},
  volume = {61},
  pages = {2430--2438},
  number = {5},
  month = {May}
}

@BOOK{Tan14,
  title = {{Asymptotic Estimates in Information Theory with Non-Vanishing Error Probabilities}},
  publisher = {Now Publishers Inc},
  year = {2014},
  author = {V. Y. F. Tan},
  volume = {11},
  series = {Foundations and Trends in Communications and Information Theory}
}

@ARTICLE{Hayashi09,
  author = {M. Hayashi},
  title = { Information spectrum approach to second-order coding rate in channel
	coding},
  journal = TIT,
  year = {2009},
  volume = {55},
  pages = {4947--4966},
  month = {Nov}
}

@article{Ahlswede1974TwoSenders,
  author  = {R. Ahlswede},
  title   = {The Capacity Region of a Channel with Two Senders
             and Two Receivers},
  journal = {The Annals of Probability},
  volume  = {2},
  number  = {5},
  pages   = {805--814},
  month   = oct,
  year    = {1974},
  doi     = {10.1214/aop/1176996549},
  url     = {https://projecteuclid.org/euclid.aop/1176996549}
}

@phdthesis{Liao1972MultipleAccess,
  author  = {Liao, H. H. J.},
  title   = {Multiple Access Channels},
  school  = {University of Hawaii},
  address = {Honolulu, HI},
  year    = {1972}
}

@inproceedings{Ahlswede1971Multiway,
  author    = {R. Ahlswede},
  title     = {Multi-Way Communication Channels},
  booktitle = {Proceedings of the Second International Symposium on
               Information Theory},
  pages     = {23--52},
  address   = {Tsahkadsor, Armenian SSR},
  publisher = {Publishing House of the Hungarian Academy of Sciences},
  year      = {1971}
}

@incollection{Shannon1961TwoWay,
  author    = {C. E. Shannon},
  title     = {Two-Way Communication Channels},
  editor    = {Neyman, Jerzy},
  booktitle = {Proceedings of the Fourth Berkeley Symposium on
               Mathematical Statistics and Probability},
  volume    = {1},
  pages     = {611--644},
  publisher = {University of California Press},
  address   = {Berkeley, CA},
  year      = {1961},
  url       = {https://projecteuclid.org/euclid.bsmsp/1200512185}
}

@ARTICLE{PPV10,
  author = {Y. Polyanskiy and H. V. Poor and S. Verd\'{u}},
  title = {Channel coding in the finite blocklength regime},
  journal = TIT,
  year = {2010},
  volume = {56},
  pages = {2307--2359},
  month = {May}
}

@BOOK{Poor,
  title = {An Introduction to Signal Detection and Estimation},
  publisher = {Springer},
  year = {1998},
  author = {H. V. Poor},
  address = {2nd}
}

@ARTICLE{VH94,
  author = {S. Verd\'{u} and T. S. Han},
  title = {A general formula for channel capacity},
  journal = TIT,
  year = {1994},
  volume = {40},
  pages = {1147--1157},
  number = {4},
  month = {Apr}
}
\end{document}